\documentclass{lmcs}
\pdfoutput=1

\usepackage{lastpage}
\lmcsdoi{20}{1}{11}
\lmcsheading{}{\pageref{LastPage}}{}{}%
{Dec.~08,~2022}{Jan.~31,~2024}{}


\theoremstyle{definition}
\newtheorem{example}[thm]{Example}
\usepackage{multicol}

\usepackage[utf8]{inputenc}
\usepackage{amsmath,
	    amsfonts,
	    amstext,
	    amssymb,
	    mathrsfs,
	    stmaryrd,
	    latexsym,
	    enumitem,
	    proof,
	    graphicx,
	    color,
	    hyperref,
	    comment}

\definecolor{shade}{rgb}{0.7,0.7,0.7}

\usepackage{hyperref}
\usepackage{amsmath}
\usepackage{amssymb}
\usepackage{mathrsfs}
\usepackage{listings}

\usepackage{wrapfig}
\usepackage{lipsum}

\usepackage{xcolor}
\usepackage{mathpartir}
\usepackage{mathtools}
\usepackage{xypic}
\usepackage{xcolor}
\usepackage[capitalise]{cleveref}

\usepackage{macro}

\usepackage{placeins}

\usepackage{thmtools}
\usepackage{thm-restate}

\begin{document}

\title[$\rtpl$: The Reversible Temporal Process Language]{$\rtpl$: The Reversible Temporal Process Language\rsuper*}

\titlecomment{{\lsuper*}This paper is a revised and extended version of~\cite{BocchiLMY22}.}

\thanks{
  This work has been partially supported by the BehAPI project funded by the EU H2020 RISE under the Marie Sklodowska-Curie action (No: 778233), by the EU HEU Marie Sklodowska-Curie action ReGraDe-CS (No: 101106046), by EPSRC project EP/T014512/1 (STARDUST), by MIUR PRIN project NiRvAna,
  by MIUR PRIN project DeKLA,
   by French ANR project DCore ANR-18-CE25-0007, by INdAM -- GNCS 2022 project \emph{Propriet\`a qualitative e quantitative di sistemi reversibili} and GNCS 2023 project \emph{Reversibilit\`a In SIstemi COncorrenti: analisi quantitative e funzionali (RISICO)}, code CUP\_E53C22001930001, and by JSPS KAKENHI Grant Number JP21H03415. We thank the anonymous referees of this paper and of its conference version for their helpful comments and suggestions.}

\author[L. Bocchi]{Laura Bocchi\lmcsorcid{0000-0002-7177-9395}}[a]
\author[I. Lanese]{Ivan Lanese\lmcsorcid{0000-0003-2527-9995}}[b]
\author[C.A. Mezzina]{Claudio Antares Mezzina\lmcsorcid{0000-0003-1556-2623}}[c] 
\author[S. Yuen]{Shoji Yuen\lmcsorcid{0000-0003-2642-0647}}[d]

\address{School of Computing, University of Kent, UK}
\email{L.Bocchi@kent.ac.uk}
\address{Focus Team, University of Bologna/INRIA, Italy}
\email{ivan.lanese@gmail.com}
\address{Dipartimento di Scienze Pure e Applicate, Universit\`a di Urbino, Italy}
\email{claudio.mezzina@uniurb.it}
\address{Graduate School of Informatics, Nagoya University, Japan}
\email{yuen@i.nagoya-u.ac.jp}

\begin{abstract}
Reversible debuggers help programmers to find the causes of misbehaviours in concurrent programs more quickly, by executing a program backwards from the point where a misbehaviour was observed, and looking for the bug(s) that caused it. Reversible debuggers can be founded on the well-studied theory of causal-consistent reversibility, which only allows one to undo an action provided that its consequences, if any, are undone beforehand. Causal-consistent reversibility yields more efficient debugging by reducing the number of states to be explored when looking backwards. 
Till now, causal-consistent reversibility has never considered time, which is a key aspect in real-world applications.
Here, we study the interplay between reversibility and time in concurrent systems via a process algebra. The Temporal Process Language (TPL) by Hennessy and Regan is a well-understood extension of CCS with discrete-time and a timeout operator. We define $\rtpl$, a reversible extension of TPL, and we show that it satisfies the properties expected from a causal-consistent reversible calculus. 
We show that, alternatively, $\rtpl$ can be interpreted as an extension of reversible CCS with time.
\end{abstract}

\maketitle

\section*{Introduction}\label{sec:intro}

Recent studies~\cite{undo,cost_deb} show that reversible debuggers ease the debugging phase, and help programmers to
quickly find the causes of a misbehaviour. Reversible debuggers can be built on top of a causal-consistent reversible semantics~\cite{GiachinoLM14,LaneseNPV18,FabbrettiLS21,LamiLSCF22}, and this approach is particularly suited to deal with concurrency bugs, which are hard to find using traditional debuggers~\cite{Gray86}. 
By exploiting causality information, causal-consistent reversible debuggers allow one to undo just the steps which led (that is, are causally related) to a visible misbehaviour, reducing the number of steps/spurious causes and helping to understand the root cause of the misbehaviour. More precisely, one can explore backwards the tree of causes of a visible misbehaviour, possibly spread among different processes, looking for the bug(s) causing it. In the last years several reversible semantics for concurrency have been developped, see, e.g.,~\cite{rccs,ccsk,rpi,rhoTCS,MelgrattiMU19,GiachinoLMT17,LaneseM20,BMezzina20,MedicMPY20,MelgrattiMP21b}. %
However, none of them takes into account {time}\footnote{The notion of time reversibility addressed in~\cite{BMezzina20} is not aimed at studying programming languages with constructs to support hard or soft time constraints, but at performance evaluation via (time-reversible) Markov chains.}.  
Time-dependent behaviour is an intrinsic and important feature of real-world concurrent systems and has many applications:
from the engineering of highways~\cite{MAURO2020808}, to the  
 manufacturing schedule~\cite{manufactoring} and to the scheduling problem for real-time operating systems~\cite{Bertolotti05}. 

Time is instrumental for the functioning of embedded systems where some events are triggered by the system clock. Embedded systems are used for both real-time and soft real-time applications, frequently in safety-critical scenarios. Hence, before being deployed or massively produced, they have to be heavily tested, and hence possibly debugged. 
Actually, debugging occurs not only upon testing, but in almost all the stages of the life-cycle of a software system: from the early stages of prototyping to the post-release maintenance (e.g., updates or security patches). Concurrency is important in embedded systems~\cite{FantGP12a}, and concurrency bugs frequently happen in these systems as well~\cite{embedded_bug}. To debug such systems, and deal with time-dependent bugs in particular, it is crucial that debuggers can handle both concurrency and time.

In this paper, we study the interplay between time and reversibility in a process algebra for concurrency. 
In the literature, there exists a variety of timed process algebras for the analysis and specification of concurrent timed systems~\cite{NicollinS91}.  We build on the Temporal Process Language (TPL)~\cite{tpl}, a CCS-like process algebra featuring an \emph{idling} prefix (modelling a delay) and a \emph{timeout} operator. The choice of TPL is due to its simplicity and its well-understood theory. We define $\rtpl$, a reversible extension of TPL, and we show that it satisfies the properties expected from a causal-consistent reversible calculus. Alternatively, $\rtpl$ can be interpreted as an extension of reversible CCS (in particular CCSK~\cite{ccsk}) with time. 

A reversible semantics in a concurrent setting is frequently defined following the causal-consistent approach~\cite{rccs,LaneseMT14} (other approaches are also used, e.g., to model biological systems~\cite{PhillipsUY12,PhilippouP18}). Causal-consistent reversibility states that any action can be undone, provided that its consequences, if any, are undone beforehand. Hence, it strongly relies on a notion of causality. To prove the reversible semantics of $\rtpl$ causal-consistent, we exploit the theory in~\cite{LanesePU20}, whereby causal-consistency follows from three key properties: 
\begin{description}
\item[Loop Lemma] any action can be undone by a corresponding backward action;
\item[Square Property]concurrent actions can be executed in any order;
\item[Parabolic Lemma] backward computations do not introduce new states.
\end{description}
The application of causal-consistent reversibility to timed systems is not straightforward, since time heavily changes the causal semantics of the language. In untimed systems, causal dependencies are either \emph{structural} (e.g., via sequential composition) or determined by \emph{synchronisations}. In timed systems further dependencies between parallel processes can be introduced by time, even when processes do not actually interact, as illustrated in the following example.
\begin{example}[Motivating example]
\label{ex:1}
Consider the following Erlang code.
\definecolor{dkgreen}{rgb}{0,0.6,0}
\definecolor{gray}{rgb}{0.5,0.5,0.5}
\definecolor{mauve}{rgb}{0.58,0,0.82}
\setlength{\columnsep}{-1cm}
\lstset{ %
  language=erlang,                
  basicstyle=\footnotesize,           
  numbers=left,                   
  numberstyle=\tiny\color{gray},  
  stepnumber=1,                   
  numbersep=5pt,                  
  showspaces=false,               
  showstringspaces=false,         
  showtabs=false,                 
  rulecolor=\color{black},        
  tabsize=2,                      
  captionpos=b,                   
  breaklines=true,                
  breakatwhitespace=true,        
  title=\lstname,                   
  keywordstyle=\color{blue},          
  commentstyle=\color{dkgreen},       
  stringstyle=\color{mauve},         
}
\begin{multicols}{2}
\begin{lstlisting}[language=erlang]
process_A() ->
	receive 
		X -> handleMsg()
		after 200 -> 
			handleTimeout()
	end	
end.
\end{lstlisting}
\columnbreak
\begin{lstlisting}[firstnumber=8]
process_B(Pid) -> 
	timer:sleep(500),
	Pid! Msg	
end.
	
PidA=spawn(?MODULE,process_A,[]),
spawn(?MODULE,process_B,[PidA]).	
\end{lstlisting}
\end{multicols}

Process \lstinline{A} (lines $1-7$) waits for a message; if a message is received within 200 ms, then process \lstinline{A} calls function \lstinline{handleMsg()}, otherwise it calls function \lstinline{handleTimeout()}.
Process \lstinline{B} (lines $8-11$) sleeps for 500ms and then sends a message to \lstinline{Pid}, where \lstinline{Pid} is a parameter of the function executed by process \lstinline{B} (line $8$). The code in line $13$ spawns an instance of process \lstinline{A} and uses its process identifier \lstinline{PidA} as a parameter to spawn an instance of process \lstinline{B} (line $14$).
The two process instances are supposed to communicate, but the timeout in process \lstinline{A} (line $4$) triggers after 200 ms, while process \lstinline{B} will only send the message after 500 ms (lines $9-10$). In this example, the timeout  rules out the execution where process \lstinline{A} communicates with process \lstinline{B}, which would be possible in the untimed scenario. Namely, an execution can become unviable because of a time dependency, without any actual interaction between the two involved processes. \finex
\end{example}
From a technical point of view, the semantics of TPL does not fit the formats for which a causal-consistent reversible semantics can be derived automatically~\cite{ccsk,LaneseM20}, and also the generalisation of the approaches developed in the literature for untimed models~\cite{rccs,rpi,rhoTCS} is not straightforward and is the objective of this work.

 The rest of the paper is structured as follows. Section~\ref{sec:informal} gives an informal overview of TPL and reversibility. Section~\ref{sec:tpl} introduces the syntax and semantics of the reversible Temporal Process Language 
 ($\rtpl$). In Section~\ref{sec:properties}, we relate $\rtpl$ to TPL and CCSK, 
 while Section~\ref{sec:causality} studies the reversibility properties of $\rtpl$. Section~\ref{sec:conc} concludes the paper and discusses related and future work. 
A formal background on CCS, TPL and CCSK (for the readers that wish a more rigorous overview than the informal one in Section~\ref{sec:informal}), as well as longer proofs and additional technical details are collected in Appendix.

This paper is an extended and revised version
of~\cite{BocchiLMY22}. The semantics has been revised since the one
in~\cite{BocchiLMY22} failed to capture some time dependencies when
going back and forward (cf.~Example~\ref{ex:nocc}). We now also provide a better characterisation of the causality model of $\rtpl$ (cf.~Proposition~\ref{prop:backalphasigma} and Theorem~\ref{th:total}).
Further technical improvements include
the formulation of the correspondences
between $\rtpl$, $\tpl$, CCSK and CCS in terms of (bi)simulations
(\cref{BehavForgetH,BehavForgetT}). We now also provide full
proofs of all results as well as additional explanations and examples.
Finally, the whole presentation has been carefully refined.

\section{Informal Overview of TPL and Reversibility}\label{sec:informal}

In this section we give an informal overview of Hennessy \& Regan's TPL (Temporal Process Language)~\cite{tpl} and introduce a few basic concepts of causal-consistent reversibility~\cite{rccs,LanesePU20}. For a more rigorous introduction, the interested reader can find the syntax and semantics of TPL in Appendix~\ref{app:tpl} and the syntax and semantics of the reversible calculus CCSK~\cite{ccsk} in Appendix~\ref{app:ccsk}. The syntax and semantics of CCS, which is at the basis of both TPL and CCSK, is in Appendix~\ref{app:ccs}.

\subsection{Overview of TPL}\label{sec:tplov} Process $\timed{pid. P}{Q}$ models a timeout: it can either immediately do action $pid$ followed by $P$ or, in case of delay, continue as $Q$. In transition (\ref{exa1}) the timeout process is in parallel with co-party $\overline{pid}.0$ that can immediately synchronise with action $pid$, and hence the timeout process continues as $P$.  
\begin{equation}\label{exa1}
\overline{pid}.0 \parallel  \timed{pid. P}{Q} \fwd{\tau} 0 \parallel P
 \end{equation}
In transition (\ref{exa2}), $\timed{pid. P}{Q}$ is in parallel with process $\sigma. \overline{pid}.0$ that can synchronise only after a delay of one time unit $\sigma$ ($\sigma$ is called a time action). Because of the delay, the timeout process continues as $Q$:
\begin{equation}\label{exa2}
 \sigma. \overline{pid}.0 \parallel \timed{pid. P}{Q}\fwd{\sigma} \overline{pid}.0 \parallel Q 
 \end{equation}
The processes on the left-hand side of transition (\ref{exa2}) describe the interaction structures of the Erlang program in Example~\ref{ex:1}. More precisely, the timeout of $200$ time units in process \lstinline{A} can be encoded using nested timeouts: 
\[A(0) = Q \quad A(n+1) = \timed{pid. P}{A(n)} \quad (n\in \naturals)\]
while process \lstinline{B}  can be modelled as the sequential composition of $500$ actions $\sigma$ followed by action $\overline{pid}$, as follows:
\[B(0) = \overline{pid} \quad B(n+1) = \sigma.{B(n)} \quad (n\in \naturals)\] 
Using the definition above, $\timed{pid. P}{A(200)}$ models a process that executes $pid$ and continues as $P$ if a co-party is able to synchronise within $200$ time units, otherwise executes $Q$.
Hence, Example~\ref{ex:1} is rendered as follows:
$$\timed{pid. P}{A(200)} \parallel B(500)$$
The design of TPL is based on (and enjoys) three properties~\cite{tpl}: time-determi\-nism, patience, and maximal progress. 
\emph{Time-determinism} means that time actions from one state can never reach distinct states, formally: if $P\fwd{\sigma}Q$ and  $P\fwd{\sigma}Q'$ then $Q=Q'$. A consequence of time-determinism is that choices can only be decided via communication actions and not by time actions, for example $\alpha.P + \beta.Q $ can change state by action $\alpha$ or $\beta$, but not by time action $\sigma$. Process $\alpha.P + \beta.Q$ can make an action $\sigma$, by a property called patience, but this action would not change the state, as shown in transition (\ref{exa3}).  
\begin{equation}\label{exa3}
\alpha.P + \beta.Q \fwd{\sigma}\alpha.P + \beta.Q
\end{equation}

\emph{Patience} ensures that communication processes like $\alpha.P$ can indefinitely delay communication $\alpha$ with $\sigma$ actions (without changing state) until a co-party is available. For example, by patience, process $\alpha.P$ in (\ref{exa31}) can delay the communication on $\alpha$ until the other process $\sigma.\overline{\alpha}.Q$ is ready to communicate:
\begin{equation}\label{exa31}
\alpha.P \parallel \sigma.\overline{\alpha}.Q \fwd{\sigma}\alpha.P \parallel \overline{\alpha}.Q\fwd{\tau}P\parallel Q
\end{equation}
\emph{Maximal progress} states that (internal/synchronisation) $\tau$ actions cannot be delayed, formally: if $P\fwd{\tau}Q$ then there is no $Q'$ such that $P\fwd{\sigma}Q'$. Namely, a delay can only be attained either via explicit $\sigma$ prefixes or because synchronisation is not possible. 
Basically,  patience allows for time actions when communication is not possible, and maximal progress disallows time actions when communication is possible:
\[\begin{array}{lll}
\alpha.P\fwd{\sigma} &  \text{ (by patience)}\\
\alpha.P \parallel \overline{\alpha}.Q\not\fwd{\sigma} \text{ because } \alpha.P \parallel \overline{\alpha}.Q\fwd{\tau} & \text{ (by maximal progress)} 
\end{array}
\]

\subsection{Overview of causal-consistent reversibility.}
Before presenting $\rtpl$, we discuss the \emph{reversing technique} we adopt.
In the literature, two approaches to define a causal-consistent extension of a given calculus or language have been proposed: \emph{dynamic} and \emph{static}~\cite{LaneseMM21}. 
The dynamic approach (as in~\cite{rccs,rpi,rhoTCS}) makes explicit use of memories to keep track of past events and causality relations, while 
the static approach (originally proposed in~\cite{ccsk}) is based on two ideas: making all the operators of the language static so that no information is lost and using communication keys to keep track of which events have been executed. In the dynamic approach, constructors of processes disappear upon transitions (as in standard calculi). 

For example, in the following CCS transition:
$$a.P \xrightarrow{a} P$$ the action $a$ disappears as an effect of the transition. The dynamic approach prescribes to use memories to keep track of the discarded items. In static approaches, such as~\cite{ccsk}, actions are syntactically maintained, and process $a.P$ can perform the transition below
$$a.P \xrightarrow{a[i]} \magenta{a[i].}P$$
where $P$ is decorated with the executed action $a$ and a unique key $i$. 
The term $\magenta{a[i].}P$ acts like $P$ in forward transitions, while the coloured part decorating $P$ is used to define backward transitions, e.g., 
$$\magenta{a[i].}P \bk{a[i]} a.P$$
Keys are important to correctly revert synchronisations. Consider the process below. It can take two forward synchronisations with keys $i$ and $j$, respectively: 
$$a.P_1 \parallel \overline{a}.P_2 \parallel a.Q_1 \parallel \overline{a}.Q_2 \xrightarrow{\tau[i]} \xrightarrow{\tau[j]} \magenta{a[i].}P_1 \parallel \magenta{\overline{a}[i].}P_2 \parallel \magenta{a[j].}Q_1 \parallel \magenta{\overline{a}[j].}Q_2$$ 
From the reached state, there are two possible backward actions: $\tau[i]$ and $\tau[j]$. The keys are used to ensure that a backward action, say $\tau[i]$, only involves parallel components that have previously synchronised and not, for instance, $\magenta{a[i].}P_1$ and $\magenta{\overline{a}[j].}Q_2$.
When looking at the choice operator, in the following CCS transition:
$$a.P + b.Q \xrightarrow{a} P$$
both the choice operator ``+'' and the discarded branch $b.Q$ disappear as an effect of the transition. 
In static approaches, the choice operator and the discarded branch are syntactically maintained, and process $a.P + b.Q$ can perform the transition below:
$$a.P + b.Q \xrightarrow{a[i]} \magenta{a[i].}P\magenta{+b.Q}$$
where $\magenta{a[i].}P\magenta{+b.Q}$ acts like $P$ in forward transitions, while the coloured part allows one to undo $a[i]$ and then possibly proceed forward with an action $b[j]$. 

In this paper, we adopt the static approach since it is simpler, while the dynamic approach is more suitable to complex languages such as the $\pi$-calculus, see the discussion in~\cite{LaneseMM21,LaneseP21}.

\section{The Reversible Temporal Process Language}\label{sec:tpl}

In this section we define $\rtpl$, an extension of Hennessy \& Regan TPL (Temporal Process Language)~\cite{tpl} with causal-consistent reversibility following the static approach in the style of~\cite{ccsk}. 
\subsection{Syntax of $\rtpl$.} We denote with $\procs$ the set of all the configurations generated by the grammar in Figure~\ref{fig:syn}. 
\begin{figure}[t]
\begin{align*}
\textrm{(Processes)}\quad & P = \,   \pi.P \sep \timed{P}{Q} \sep  P + Q \sep P \parallel Q \sep \res{P}{a} \sep A \sep \nil \\[0.2cm]
\textrm{(Configurations)}\quad  &  X = \rho\colorkey{i}.X \sep\timedr{X}{Y}{i} \sep  \timedl{X}{Y}{i} \mid X + Y  \sep \\ 
&\,\qquad X \parallel Y \sep X\setminus a \sep P\\
\textrm{(Communication actions)}\quad 
& \alpha =\, a \sep \co{a} \sep \tau\\
\textrm{(Prefixes)}\quad 
& \pi = \, \alpha \sep \sigma\\
\textrm{(Runtime prefixes)}\quad 
& \rho =\, \pi \sep \sigma_{\bot} 
\end{align*}
\caption{Syntax of $\rtpl$}
\label{fig:syn}
\end{figure}

\textit{Processes} ($P, Q, \ldots$) describe timed interactions following~\cite{tpl}. We let $\names$ be the set of action names $a$, $\co{\names}$ the set of action conames $\co{a}$. We use $\alpha$ to range over $a$, $\co{a}$ and internal actions $\tau$. We assume $\co{\co{a}}=a$. In process $\pi.P$, prefix $\pi$ can be a communication action $\alpha$  or a time action $\sigma$, and $P$ is the continuation. Timeout $\timed{P}{Q}$ executes either $P$ (if possible) or $Q$ (in case of timeout). 
$P+Q$, $P\parallel Q$, $\res{P}{a}$, $A$, and $\nil$ are the usual choice, parallel composition, 
name restriction, recursive call, and terminated process from CCS.  For each recursive call $A$ we assume a recursive definition $A \stackrel{def}{=}P$. We also assume recursion to be guarded, hence recursive variables can only occur under prefix.

\textit{Configurations} ($X, Y, \ldots$) describe states via annotation of executed actions with keys following the static approach. We let $\keys$ be the set of all keys ($k, i, j, \ldots$).  Configurations are processes with (possibly) some computational history (i.e., prefixes marked with keys): 
$\pi\colorkey{i}.X$ is the configuration that has already executed $\pi$, and the execution of such $\pi$ is identified by key $i$. Configuration $\timedl{X}{Y}{i}$ is executing the main branch $X$ whereas $\timedr{X}{Y}{i}$ is executing $Y$. 
Some TPL processes, namely patient processes like $\alpha.P$ illustrated earlier in (\ref{exa31}), allow time to pass without changing their own structure. This is an issue in $\rtpl$, since it may lead different parallel components to have a different understanding of the passage of time, while we want time to pass at the same pace for each parallel component. For this reason, to record that time has passed for a patient process, we use a special prefix $\ghost{i}.X$. Namely, $\ghost{i}.X$ is the configuration which has patiently registered the passage of time along with key $i$. Prefix $\ghost{i}$ differs from $\sigma\colorkey{i}$ since the former is only for the current execution (patient delays \emph{may} happen but do not always have to), while the latter requires time to pass in each possible execution (see Example~\ref{ex:ghostright}). We will discuss this issue in more detail in \cref{sec:causality}. We use $\rho$ to denote either $\pi$ or $\ghostlab$. 

A configuration can be thought of as a context with actions that have already been executed, each associated to a key, containing a process $P$, with actions yet to execute and hence with no keys. Notably, keys are distinct but for actions happening together: an action and a co-action that synchronise, or the same time action traced by different processes, e.g., by two parallel delays. 
A configuration $P$ can be thought of as the initial state of a computation, where no action has been executed yet. We call such configurations \emph{standard}. Definition~\ref{def:key} formalises this notion via function $\key(X)$ that returns the set of keys of a given configuration. 

\begin{defi}[Standard configuration]\label{def:key}
The set of keys of a configuration $X$, written $\key(X)$, is inductively defined as follows:
\[\begin{array}{llllll}
& \key(P) = \emptyset &&\quad \key(\rho\colorkey{i}.X) = \{i\} \cup \key(X) &&\quad \key(\res{X}{a}) = \key(X) \\
&\multicolumn{5}{l}{\key(\timedr{Y}{X}{i}) = \key(\timedl{X}{Y}{i}) = \{i\} \cup \key(X)}\\
& \multicolumn{5}{l}{\key(X + Y) = \key(X\parallel Y) = \key(X) \cup \key(Y)}\\
\end{array}\]
A configuration $X$ is \emph{standard}, written $\std(X)$, if $\key(X)=\emptyset$. 
\end{defi}
Basically, a standard configuration is a process. 
To handle the delicate interplay between time-determinism and reversibility of time actions, it is useful to distinguish the class of configurations that 
have not executed any \emph{communication} action (but may have executed time actions). We call these configurations \emph{not-acted} and characterise them formally using the predicate $\nact(\cdot)$ below.

\begin{defi}[Not-acted configuration]\label{def:nact} The not-acted predicate $\nact(\cdot)$ is inductively defined as:
\[
	\begin{array}{ll}
	&\nact(\nil)	= \nact(A) 	= \nact (\timed{X}{Y})  = \nact(\pi.X) = \true \\
	& \nact(\alpha\colorkey{i}.X) = \nact(\timedl{X}{Y}{i}) = \false   \\
	& \nact(\sigma\colorkey{i}.X) =\nact(\ghost{i}.X) =  \nact(\res{X}{a}) = \nact(\timedr{Y}{X}{i})   = \nact(X) \\
	& \nact(X\parallel Y) =  \nact(X + Y) = \nact(X) \wedge \nact(Y)
\end{array}
\]
A configuration $X$ is \emph{not-acted} (resp. \emph{acted}) if $\nact(X)=\true$ (resp. $\nact(X)=\false$). 
\end{defi}
Basic standard configurations are always not-acted (first line of Definition~\ref{def:nact}). Indeed, it is not possible to reach a configuration $\pi.X$ where $X$ is acted. 
In the second line, a configuration that has executed communication actions is acted. In particular, we will see that $\timedl{X}{Y}{i}$ is only reachable via a communication action.  
The configurations in the third line are not-acted if their continuations are not-acted. 
For parallel composition and choice, $\nact(\cdot)$ is defined as a conjunction. For example $\nact(\alpha\colorkey{i}.P \parallel \beta.Q)=\false$ and $\nact(\alpha\colorkey{i}.P + \beta.Q)=\false$. Note that in a choice configuration $X+Y$, at most one between $X$ and $Y$ can be acted.
Whereas $\std(X)$ implies $\nact(X)$, the opposite implication does not hold. For example, $\std(\sigma\colorkey{i}.\nil)=\false$ but $\nact(\sigma\colorkey{i}.\nil)=\true$.

\subsection{Semantics of $\rtpl$.}

We denote with $\act$ the set $\names \cup \co{\names} \cup \{\tau,\sigma\}$ of actions and let $\pi$ to range over the set $\act$. We define the set of all the labels $\lbl= \act \times  \keys $.
 The labels associate each $\pi\in \act$ to a key $i$. 
 The key is used to associate the forward occurrence of an action with its corresponding reversal. Also, instances of actions occurring together (synchronising action and co-action or the effect of time passing in different components of a process) have the same key, otherwise keys are distinct.
 \begin{figure}
\input{forwardLTS}
\caption{$\rtpl$ forward LTS}
\label{fig:fw}
\end{figure}

\begin{defi}[Semantics]
The operational semantics of $\rtpl$ is given by two Labelled Transition Systems (LTSs) defined on  the same set of all configurations 
$\procs$, and the set of all labels $\lbl$: 
a forward LTS ($\procs$, $\lbl$, $\fwd{}$) and 
a backward LTS ($\procs$, $\lbl$, $\bk{}$).  We define $\red{} = \fwd{} \cup \bk{}$,
where $\fwd{}$ and $\bk{}$ are the least transition relations induced by the rules in Figure~\ref{fig:fw} and Figure~\ref{fig:bk}, respectively.

\end{defi}
Given a relation $\mathcal{R}$, we indicate with 
 $\mathcal{R}^{*}$ its reflexive and transitive closure. 
We use notation $X\not\fwd{\tau}$ (resp.  $X\not\bk{\tau}$) when there are no configuration $X'$ and key $i$ such that $X \fwd{\tau[i]}X'$ (resp. $X\bk{\tau[i]}X'$). 

We now discuss the rules of the forward semantics (Figure~\ref{fig:fw}).
Rule [\textsc{PAct}] describes patient actions: in TPL process $\alpha.P$ can make a time step to itself. This kind of actions allows a process to wait indefinitely until it can communicate (by patience~\cite{tpl}). However, in $\rtpl$ we need to track passage of time, hence rule [\textsc{PAct}] adds a $\ghost{i}$ prefix in front of the configuration, with a key $\ckey{i}$.
Rule [\textsc{RAct}] executes actions $\alpha\colorkey{i}$ or $\sigma\colorkey{i}$ on a prefix process. Observe that, unlike patient time actions on $\alpha.P$, a time action on $\sigma.P$ corresponds to a deliberate and planned time consuming action and, therefore, it executes the $\sigma$ prefix, hence no $\ghostlab$ prefix needs to be added. 
Rule [\textsc{Idle}] registers passage of time on a $\nil$ configuration by adding a $\ghost{i}$ prefix to it.
Rule [\textsc{Act}] lifts actions of the continuation $X$ on configurations where prefix $\rho\colorkey{i}$ has already been executed. Side condition $\ckey{j} \neq \ckey{i}$ ensures freshness of $\ckey{j}$ is preserved.  
Rules [\textsc{STout}] and [\textsc{SWait}] model timeouts. In rule [\textsc{STout}], if $X$ is not able to make $\tau$ actions then $Y$ is executed; this rule models a timeout that triggers only if the main configuration $X$ is stuck. The negative premise on [\textsc{Stout}] can be encoded into a decidable positive one as shown in  Appendix~\ref{sec:negative}. In rule [\textsc{Tout}] instead the main configuration can execute and the timeout does not trigger. Rule [\textsc{SWait}] (resp. [\textsc{Wait}]) models transitions inside a timeout configuration where the $Y$ (resp.~$X$) branch has been previously taken. The semantics of timeout construct becomes clearer in the larger context of parallel configurations, when looking at rule [\textsc{SynW}]. 
Rule [\textsc{SynW}] models time passing for parallel configurations. The negative premise 
ensures that, in case $X$ or $Y$ is a timeout configuration, timeout can trigger only if no synchronisation may occur, that is if the configurations are stuck. [\textsc{SynW}] requires time to pass in the same way (an action $\sigma$ is taken by both components, with the same key $\ckey{i}$) for the whole system.
Rules [\textsc{Par}] (and symmetric) and [\textsc{Syn}] are as usual for communication actions and allow parallel configurations to either proceed independently or to synchronise. In the latter case, the keys need to coincide.
Defining the semantics of choice configuration $X+Y$ requires special care to ensure time-determinism (recall, choices are only decided via communication actions). Also, we need to record time actions  to be able to reverse them correctly (cfr.~Loop Lemma, discussed later on in Lemma~\ref{lem:loop}).
Rule [\textsc{ChoW}] describes the passage of time along a choice configuration $X+Y$.
Since time does not decide a choice, both branches have to execute the same time action with the same key. 
Rule [\textsc{Cho}] allows one to take one branch, or continue executing a previously taken branch. The choice construct is syntactically preserved, to allow for reversibility, but the one branch that is not taken remains non-acted (i.e., $\nact(Y)$). This ensures that choices can be decided by a communication action only.
Let us note that even in case of a decided choice, that is a choice configuration in which one of the two branches has performed a communication action, time actions are registered by both configurations. For example, the configuration $a.\nil+\sigma.\nil$  can execute the following transitions:
\[a.\nil+\sigma.\nil\fwd{a\colorkey{i}} a\colorkey{i}.\nil +\sigma.\nil \fwd{\sigma\colorkey{j}}
a\colorkey{i}.\ghost{j}.\nil +\sigma\colorkey{j}.\nil
\]
After the $a\colorkey{i}$ action, even if the left branch of the choice has been selected, both branches participate to the time action $\sigma\colorkey{j}$.
Rules [\textsc{Hide}] and [\textsc{Const}] are standard. 

The rules of the backward semantics, in Figure~\ref{fig:bk}, undo communication and time actions executed under the forward semantics.
 Backward rules are symmetric to the forward ones.

Now that we have introduced both the forward and the backward rules we can clarify the difference between $\ghostlab$ and $\sigma$.

\begin{example}\label{ex:ghostright}
Let us consider the patient process $a.P$. We can have the following derivation:
\begin{equation}
a.P \fwd{\sigma\colorkey{i}} \ghost{i}.a.P \bk{\sigma\colorkey{i}} a.P \fwd{a\colorkey{j}} a\colorkey{j}.P
\end{equation}
where $a.P$ executes forwards in two different ways: first by letting time pass, later on by interacting on $a$. Notice that for these interactions to be possible in a larger context we need the context to have changed as well.

We can try to have a similar derivation using process $\sigma.a.P$ instead, but the final outcome is not the same:
\begin{equation}\label{ex:ghost}
\sigma.a.P \fwd{\sigma\colorkey{i}} \sigma\colorkey{i}.a.P \bk{\sigma\colorkey{i}} \sigma.a.P
\end{equation}
Indeed, at this stage $\sigma.a.P$ cannot interact on $a$. In general,
$\sigma$ requires time to pass in every possible computation, while
$\ghostlab$ does not.\finex
\end{example}

\begin{figure}[!t]
\input{backwardLTS}
\caption{$\rtpl$ backward LTS}
\label{fig:bk}
\end{figure}

\begin{defi}[Reachable configurations]
A configuration $\conf{T}{X}$ is reachable 
if there exist a process $\conf{0}{P}$ and a derivation $\conf{0}{P} \red{}^* \conf{T}{X}$.
\end{defi}
Basically, a configuration is
reachable if it can be obtained via forward and backward actions from a standard configuration.

\section{Relations with TPL and reversible CCS}\label{sec:properties}

\label{sec:corr}

We can consider $\rtpl$ as a reversible extension of TPL, but also as an extension of reversible CCS (in particular CCSK~\cite{ccsk}) with time.
First, if we consider the forward semantics only, then we have a tight correspondence with TPL. To show this we define a forgetful map which discards the history information of a configuration.

\begin{defi}[History forgetting map]
The \emph{history forgetting map} $\forgeth: \mathcal{X}\rightarrow \mathcal{P}$ is inductively defined as follows: 
\[\begin{array}{ll}
\forgeth(P) = P & \forgeth(\rho\colorkey{i}.X) = \forgeth(X)\\[2mm]
\forgeth(\timedl{X}{Y}{i})= \forgeth(X) & \forgeth(\timedr{X}{Y}{i}) = \forgeth(Y)\\[2mm]
\forgeth(X \parallel Y) = \forgeth(X) \parallel \forgeth(Y) \hspace{1.5cm}& \forgeth(\res{X}{a}) = \res{\forgeth(X)}{a}\\[2mm]
\multicolumn{2}{l}{\forgeth(X + Y) = \left\{ 
					\begin{array}{ll}\forgeth(X) &\text{ if } \neg\nact(X)  \wedge \nact(Y)   \\
					\forgeth(Y) &\text{ if } \neg\nact(Y)  \wedge \nact(X) \\
					 \forgeth(X) + \forgeth(Y) & \text{ otherwise}
		\end{array}\right.}
\end{array}\]
\end{defi}
The definition above deletes all the
information about history from a configuration $X$, hence it is the identity on standard configuration $P$. Even more, each configuration is mapped into a standard one. Notice that in a non-standard timeout, only the chosen branch is taken.
In TPL
 time cannot decide choices. This is reflected into the
definition of $\forgeth(X + Y)$,  where a branch disappears only if the other one did at least a communication action.

Notably, the restriction of $\forgeth$ to untimed configurations (namely configurations containing neither timeouts nor $\sigma$ prefixes nor $\ghostlab$ prefixes) is a map from CCSK~\cite{ccsk} to CCS.
Following the notation of Appendix \ref{sec:bg}, 
we will indicate with $\rightarrow_{\mathtt{t}}$ the semantics of TPL~\cite{tpl}, reported in
Appendix \ref{app:tpl},  and with
$\mapsto_{\mathtt{k}}$ the semantics of CCSK~\cite{ccsk},  reported in
Appendix \ref{app:ccsk}.

\begin{prop}[Embedding of TPL]
\label{prop:totpl}
Let $X$ be a reachable $\rtpl$ configuration:
\begin{enumerate}
	\item if $X\fwd{\pi\colorkey{i}} Y$ then $\forgeth(X) \xrightarrow{\pi}_{\mathtt{t}} \forgeth(Y)$;
	\item if $\forgeth(X) \xrightarrow{\pi}_{\mathtt{t}} Q$ then for any $i\in \keys \setminus \key(X)$ there is $Y$ such that $X\fwd{\pi\colorkey{i}} Y$ with $\forgeth(Y)=Q$.

\end{enumerate}
\end{prop}
\begin{proof}
\begin{description}
\item[(1)] by induction on the derivation $X\fwd{\pi\colorkey{i}} Y$, with a case analysis on the last applied rule. We detail a few sample rules.

 If the move is by rule [\textsc{PAct}] then we have $Y=\ghost{i}.\alpha.X_1$, 
  with $\forgeth(X) = \forgeth(Y)$,
  and in TPL we have a corresponding state-preserving move with label $\sigma$ 
  derived using rule \textsc{Act}$_2$ in Figure~\ref{fig:tplsem2}.
 
  In the case of rule [\textsc{Act}], 
  $X=\rho\colorkey{i}.Z$ and $Y= \rho\colorkey{i}.Z'$ with  $Z\fwd{\pi\colorkey{j}}Z'$. By inductive hypothesis, in TPL $\forgeth(Z) \xrightarrow{\pi}_{\mathtt{t}} \forgeth(Z')$. Since $\forgeth(X)=\forgeth(Z)$ and $\forgeth(Y) = \forgeth(Z')$ we are done. The cases for [\textsc{Const}] and [\textsc{Hide}] are similar by induction. The cases for [\textsc{SynW}] and [\textsc{ChoW}] follow by induction as well.  

If the last applied rule is
[\textsc{Cho}], then we have that $X= X_1 + X_2$ with $X_1\fwd{\alpha\colorkey{i}} X'_1$ and $\nact(X_2)$. Also, $Y = X'_1 + X_2$ with $X'_1$ acted. Hence, $\forgeth(Y)=\forgeth(X'_1+X_2)=\forgeth(X'_1)$. We consider the case $\nact(X_1)$, the other one is simpler. By definition, we have that $\forgeth(X_1+X_2) = \forgeth(X_1) +\forgeth(X_2).$ Since $X_1\fwd{\alpha\colorkey{i}} X'_1$, by applying the inductive hypothesis we also have that
$\forgeth(X_1)\fwd{\alpha}_{\mathtt{t}} \forgeth(X'_1)$. Also, since the label is not a $\sigma$ action, in TPL 
we can use 
rule \textsc{Sum}$_1$ in Figure \ref{fig:tplsem2}   
that from $\forgeth(X_1)\fwd{\alpha}_{\mathtt{t}} \forgeth(X'_1)$ allows one to derive $\forgeth(X_1) + \forgeth(X_2)\fwd{\alpha} \forgeth(X'_1)$, as desired.

If the last applied rule is [\textsc{STout}] then $\timed{X_1}{X_2} \fwd{\sigma\colorkey{i}}\timedr{X_1}{X_2}{i}$ with $X_1$ and $X_2$ standard, and $X_1$ that can not perform $\tau$ steps. By inductive hypothesis $\forgeth(X_1)$ can not perform $\tau$ steps in TPL, hence in TPL
we can use
rule \textsc{THEN}$_2$ in Figure \ref{fig:tplsem2}  to derive
$\forgeth(\timed{X_1}{X_2}) = \timed{X_1}{X_2} \fwd{\sigma}_{\mathtt{t}} X_2 =\forgeth(\timedr{X_1}{X_2}{i})$ as desired.

\item[(2)] by induction on the definition of $\forgeth(X)$  (structural induction on $X$).
Let us first assume $X$ standard, hence $\forgeth(X) = X$.
Let us consider $X = \alpha.P$. 
In TPL, $\alpha.P$ can make a state preserving transition $\sigma$ and the corresponding $\rtpl$ configuration can match it: $X\fwd{\sigma\colorkey{i}}\ghost{i}.X$, with $\forgeth(\ghost{i}.X)=X$. Alternatively, in TPL, $\alpha.P \fwd{\alpha}_{\mathtt{t}} P$. The thesis follows since $\alpha.P \fwd{\alpha\colorkey{i}}\alpha\colorkey{i}.P$ with $\forgeth(\alpha\colorkey{i}.P)= P$. The other cases are similar, but using the induction hypothesis. 

Let us now assume $X$ non standard.
The most interesting case is when $X = X_1 +X_2$.
Let us consider $X_1$ acted (the case where $X_2$ is acted is symmetric).
In this case $\forgeth(X_1+X_2) = \forgeth(X_1)$ hence the thesis follows by inductive hypothesis using rule [\textsc{ChoW}] for $\sigma$ actions and [\textsc{Cho}] for communication actions.
If both $X_1$ and $X_2$ are not acted, then $\forgeth(X_1+X_2) = \forgeth(X_1) + \forgeth(X_2)$. We now have two cases, either $\pi = \sigma$ or $\pi=\alpha$.
If $\pi = \sigma$ we have that by 
rule \textsc{SUM}$_3$ in Figure \ref{fig:tplsem2} 
$\forgeth(X_1) \fwd{\sigma}_{\mathtt{t}} Z_1$ and  $\forgeth(X_2) \fwd{\sigma}_{\mathtt{t}} Z_2$ allow one to derive $\forgeth(X_1)+\forgeth(X_2) \fwd{\sigma}_{\mathtt{t}} Z_1+Z_2$. By inductive hypotheses we have that there exist $X'_1$ and $X'_2$ such that $X_1\fwd{\sigma\colorkey{i}}X'_1$ and 
 $X_2\fwd{\sigma\colorkey{i}}X'_2$ with $\forgeth(X'_1) = Z_1$ and  $\forgeth(X'_2) = Z_2$. We can then apply
 rule [\textsc{ChoW}] to derive $X_1+X_2\fwd{\sigma\colorkey{i}} X'_1 + X_2'$. Since $X'_1$ and $X'_2$ are still not-acted we can conclude by noticing that
 $\forgeth(X'_1 + X'_2) = \forgeth(X'_1) + \forgeth(X'_2) = Z_1 + Z_2$. 
 The case for $\pi=\alpha$ is similar. \qedhere
 \end{description}
\end{proof}

We can describe the above correspondence between $\rtpl$ and TPL in a more abstract way by adapting the notion of (strong) timed  bisimulation~\cite{timeabstracted-bisim} to relate configurations from two calculi.
\begin{defi}[Timed bisimulation]\label{def:bisim}
A binary relation $\mathcal{R}$ on $\mathcal{X} \times \mathcal{P}$
is a strong timed bisimulation between $\rtpl$ and TPL if $(X,P)\in \mathcal{R}$ implies that
\begin{enumerate}
\item if $X\fwd{\pi\colorkey{i}} Y$, then there exists $Q$ such that   $P\xrightarrow{\pi}_{\mathtt{t}}Q$ and $(Y,Q)\in \mathcal{R}$;

\item if $P\xrightarrow{\pi}_{\mathtt{t}}Q$, then there exist $Y$ and $i$ such that $X\fwd{\pi\colorkey{i}} Y$ and $(Y,Q)\in \mathcal{R}$.

\end{enumerate}
The
largest strong timed bisimulation is called strong timed equivalence, denoted $\sim$.
\end{defi}

We can now relate $\rtpl$ and TPL as follows:
\begin{thm}\label{BehavForgetH}
For each reachable $\rtpl$ configuration $X$ we have that 
$X \sim \forgeth(X)$.
\end{thm}
\begin{proof}
It is sufficient to show that the relation $\mathcal{R} = \{(X,P) \mid \forgeth(X) = P\}$ is a strong timed bisimulation.
Let us check the conditions. If $X\fwd{\pi\colorkey{i}}Y$ then thanks to
Proposition~\ref{prop:totpl} we have that
$\forgeth(X)\fwd{\pi}_{\mathtt{t}} \forgeth(Y)$ with $ \forgeth(Y) = Q$, and we have that
$(Y,Q)\in \mathcal{R}$.  If $P\fwd{\pi}_{\mathtt{t}} Q$, thanks to 
Proposition~\ref{prop:totpl} we have that $X\fwd{\pi\colorkey{i}}
Y$ with $\forgeth(Y) = Q$,
and we have that 
$(Y,Q)\in \mathcal{R}$, as desired.
\end{proof}

Also, TPL is a conservative extension of CCS. This is stated in~\cite{tpl}, albeit not formally proved. Hence, we can define a \emph{forgetful} map which discards all the temporal operators of a TPL term and get a CCS one. We can obtain a stronger result and relate $\rtpl$ with CCSK~\cite{ccsk}. That is, if we consider the untimed part of $\rtpl$ what we get is a reversible CCS which is exactly CCSK. To this end, we define a time forgetting map $\forgett$. We denote with $\mathcal{X}^k$ the set of untimed reversible configurations of $\rtpl$, which coincides with the set of all CCSK configurations (which is defined in Appendix~\ref{app:ccsk}). The set inclusion $\mathcal{X}^k \subset \mathcal{X}$ holds. 

\begin{defi}[Time forgetting map]
The \emph{time forgetting map} $\forgett: \mathcal{X}\rightarrow \mathcal{X}^k$ is inductively defined as follows: 
\[
\begin{array}{lll}
\forgett(\nil) = \nil   &&  \forgett(A) =  A
 \\
 \forgett(\alpha.P) = \alpha.\forgett(P) &&  \forgett(\alpha\colorkey{i}.X) = \alpha\colorkey{i}.\forgett(X) \\
 \forgett(X+Y) = \forgett(X) + \forgett(Y) &&\forgett(X \parallel Y) =\forgett(X) \parallel \forgett(Y) \\
   \forgett(\res{X}{a}) =  \res{\forgett(X)}{a}   && \forgett(\timed{X}{Y}{}) =\forgett(X) +  \forgett(Y)   \\
 \forgett(\sigma.P) = \forgett(P) && \forgett(\sigma\colorkey{i}.X) = \forgett(\ghost{i}.X) = \forgett(X)\\
  \forgett(\timedl{X}{Y}{i}) = \forgett(X) + \forgett(Y) && \forgett(\timedr{X}{Y}{i}) = \forgett(X) + \forgett(Y) 
\end{array}
\]
\end{defi}
Notably, the restriction of $\forgett$ to standard configurations is a map from TPL to CCS.

The most interesting aspect in the definition above is that the timeout operator $\timed{X}{Y}$ is rendered as a sum. This also happens for the decorated configurations $\timedl{X}{Y}{i}$ and $\timedr{X}{Y}{i}$. We will further discuss this design decision after Proposition~\ref{prop:rtpl_to_ccsk}.
Also, since we are relating a timed semantics with an untimed one (CCSK), the $\sigma$ actions performed by the timed semantics are not reflected in CCSK. 

\begin{prop}[Embedding of CCSK~\cite{ccsk}]
\label{prop:rtpl_to_ccsk}
Let $X$ be a reachable $\rtpl$ configuration. We have:
\begin{enumerate}
	\item if $X\fwd{\alpha\colorkey{i}} Y$ then $\forgett(X) \fwd{\alpha\colorkey{i}}_{\mathtt{k}} \forgett(Y)$;
	\item if $X\bk{\alpha\colorkey{i}} Y$ then $\forgett(X) \bk{\alpha\colorkey{i}}_{\mathtt{k}} \forgett(Y)$;        
	\item if $X\red{\sigma\colorkey{i}} Y$ then $\forgett(X) = \forgett(Y)$.
\end{enumerate}
\end{prop}
\begin{proof}
\begin{description}
\item[(1)] by induction on the derivation $X\fwd{\alpha\colorkey{i}} Y$, with a case analysis on the last applied rule.
The proof goes along the lines of the proof of Proposition~\ref{prop:totpl}, using the rules 
reported in Figures \ref{fig:ccskfwlts} and \ref{fig:ccskbklts} in Appendix~\ref{app:ccsk}.
\item[(2)] similar to the case above, using $X\bk{\alpha\colorkey{i}} Y$ instead of $X\fwd{\alpha\colorkey{i}} Y$.

\item[(3)] by induction on the derivation of $X\fwd{\sigma\colorkey{i}} Y$ (the case of backward transitions is analogous), with a case analysis on the last applied rule. Basic cases are (i) rules [\textsc{Pact}] and [\textsc{Idle}], creating a $\ghostlab$, (ii) $\sigma$ prefixes, and (iii) timeouts. In case (i) we have $Y=\ghost{i}.X$, hence $\forgett(Y) = \forgett(X)$. In case (ii) we have $\sigma.P \fwd{\sigma\colorkey{i}} \sigma\colorkey{i}.P$, hence $\forgett(\sigma.P)=P=\forgett(\sigma\colorkey{i}.P)$. In case (iii) we have $\timed{X_1}{X_2} \fwd{\sigma\colorkey{i}} \timedr{X_1}{X_2}{i}$ with $\forgett(\timed{X_1}{X_2})=\forgett(X_1)+\forgett(X_2)=\forgett(\timedr{X_1}{X_2}{i})$ as desired. Inductive cases follow by inductive hypothesis. \qedhere
\end{description}
\end{proof}

Notably, it is not always the case that transitions of the underlying untimed configuration can be matched in a timed setting. Think, e.g., of the Erlang program in Example~\ref{ex:1} (and its formalisation in Section~\ref{sec:tplov}), for a counterexample. Indeed, in the example, the communication between the two processes \lstinline|A| and \lstinline|B| is allowed in the underlying untimed model, but ruled out by incompatible timing constraints.
Also, let us consider the simple process $X = \timed{a}{b}$ and its untimed version $\forgett(X) = a+b$. The untimed process can execute the right branch as follows 
$$\forgett(X) \fwd{b\colorkey{i}} a + b\colorkey{i}$$
To match this action, $X$ has first to perform a time action and only afterwards it can take the $b$ action, as follows:
$$X\fwd{\sigma\colorkey{j}} \timedr{a}{b}{j} \fwd{b\colorkey{i}} \timedr{a}{b\colorkey{i}}{j}$$
Moreover, there are also cases where actions cannot be matched, not even after time actions. Indeed, the timeout operator $\timed{P}{Q}$ acts as a choice with left priority. For example, let us consider the process 
$X= \timed{(a \parallel \co{a})}{b}$. We have that $\forgett(X)$ can perform the $b$ action as follows  
\[\forgett(X) = (a\parallel \co{a}) + b \fwd{b\colorkey{i}} (a\parallel \co{a}) + b\colorkey{i}\]
but this action can never be matched by $X$, as $\rtpl$ maximal progress forces the internal synchronisation over time passage. Hence we can only apply rule \textsc{STout} in Figure~\ref{fig:fw}:
\[\timed{(a \parallel \co{a})}{b} \fwd{\tau\colorkey{i}} \timedl{a\colorkey{i} \parallel \co{a}\colorkey{i}}{b}{i}\]
and the resulting configuration cannot execute $b$.

Due to the examples above, we cannot characterise the relation between $X$ and $\forgett(X)$ as a bisimulation, as for $\forgeth$, but we can only prove that $\forgett(X)$ simulates $X$.
As before, we need to modify notions of simulation for reversible configurations from the literature (e.g., \cite{NicolaMV90,LaneseP21}) to relate configurations from two calculi, and to keep time into account.

\begin{defi}[Back and forward simulation]\label{def:sim}
A binary relation $\mathcal{R}$ on $\mathcal{X} \times\mathcal{X}^{k}$
is a back and forward simulation if $(X,R)\in \mathcal{R}$ implies that
\begin{enumerate}
\item if $X\fwd{\alpha\colorkey{i}} Y$, then there exists $S$ such that $R\xrightarrow{\alpha\colorkey{i}}_{\mathtt{k}}S$ and $(Y,S)\in \mathcal{R}$;
\item if $X\bk{\alpha\colorkey{i}} Y$, then there exists $S$ such that $R\bk{\alpha\colorkey{i}}_{\mathtt{k}}S$ and $(Y,S)\in \mathcal{R}$;
\item if $X\red{\sigma\colorkey{i}} Y$, then $(Y,R)\in \mathcal{R}$.
\end{enumerate}
\end{defi}
The
largest back and forward simulation is denoted by  $\precsim$.
\begin{thm}\label{BehavForgetT}
For each reachable $\rtpl$ configuration $X$ we have that 
$X \precsim \forgett(X)$.
\end{thm}
\begin{proof}
It is sufficient to show that the relation
$\mathcal{R} = \{(X,R) \mid   \forgett(X) = R\}$
is a back and forward simulation. It is easy to check the conditions of Definition~\ref{def:sim} using Proposition~\ref{prop:rtpl_to_ccsk}.
\end{proof}

\begin{figure}
\[
\xymatrix@C=1.5cm
  {
&\rtpl \ar[dl] _{\forgett}^{Prop.~\ref{prop:rtpl_to_ccsk}} \ar[dr] ^{\forgeth}_{Prop.~\ref{prop:totpl}} &\\
CCSK \ar[dr]^{\forgeth}_{Prop.~\ref{prop:totpl}}& &TPL \ar[dl]_{\forgett}^{Prop.~\ref{prop:rtpl_to_ccsk}}  \\
& CCS &
}
\]
\caption{Forgetting maps.}
\label{fig:commute}
\end{figure}

Figure~\ref{fig:commute} summarises our results:  if we remove the timed behaviour from a $\rtpl$ configuration we get a CCSK term, with the same behaviour apart for timed aspects, thanks to
Proposition~\ref{prop:rtpl_to_ccsk}. On the other side, if from $\rtpl$ we remove history information we get a TPL term (matching its forward behaviour thanks to Proposition~\ref{prop:totpl}). 
Note that the same forgetful maps (and properties) justify the arrows in the bottom part of the diagram, as discussed above. This is in line with Theorem 5.21 of~\cite{ccsk}, showing that by removing reversibility and history information from CCSK we get CCS.
Notably the two forgetting maps commute.
\begin{prop}\label{prop:commute}
For each reachable $\rtpl$ configuration $X$ we have $\forgeth(\forgett(X))=\forgett(\forgeth(X))$.
\end{prop}
\begin{proof}
By structural induction on $X$.
\end{proof}

\section{Reversibility in $\rtpl$}\label{sec:causality}
\label{sec:loop}
In a fully reversible calculus any computation can be undone.
 This is a fundamental property of reversibility~\cite{rccs,LanesePU20}, called the Loop Lemma,
and $\rtpl$ enjoys it. Formally:%

\begin{lem}[Loop Lemma]
\label{lem:loop}
If $X$ is a reachable $\rtpl$ configuration, then
$X \fwd{\pi\colorkey{i}} X' \iff X' \bk{\pi\colorkey{i}} X$
\end{lem}
\begin{proof}

We have two directions. The forward one trivially holds, since for each
forward rule of
Figure~\ref{fig:fw} there exists a symmetric one backwards in Figure~\ref{fig:bk}. The backward case
requires more attention, and we proceed by induction on $X \bk{\pi\colorkey{i}} Y$, with
a case analysis on the last applied rule. We can further distinguish the cases according 
to whether $\pi = \sigma$, $\pi = \tau$ or $\pi = \alpha$. We will just consider one instance of each case, the others are similar.
\begin{description}
	\item[$\pi = \sigma$] using rule [\textsc{ChoW}] we have that $X = X_1 + X_2$,
	 $X_1\bk{\sigma{\colorkey{i}}} X'_1$, $X_2\bk{\sigma{\colorkey{i}}} X'_2$. 
       Since by reachability of $X$ we have the reachability of $X_1$ and $X_2$, by inductive hypothesis we have that
	 $X'_1\fwd{\sigma{\colorkey{i}}} X_1$ and $X'_2\fwd{\sigma{\colorkey{i}}} X_2$, so we can apply the forward
	 version of rule [\textsc{ChoW}], as desired.
	 \item [$\pi = \tau$] using rule [\textsc{Syn}] by hypothesis we have that $X= X_1 \parallel X_2$ with
	 $X_1\bk{\alpha\colorkey{i}} X'_1$ and $X_2\bk{\co\alpha\colorkey{i}} X'_2$. 
	 By applying the inductive hypothesis we get
$X'_1\fwd{\alpha\colorkey{i}} X_1$ and $X'_2\fwd{\co\alpha\colorkey{i}} X_2$, and we can derive 
$X'_1 \parallel X'_2 \fwd{\tau\colorkey{i}} X_1 \parallel X_2$, as desired.
\item[$\pi = \alpha$] using rule [\textsc{Cho}] by hypothesis we have that $X= X_1 +X_2$,
$X_1\bk{\alpha\colorkey{i}} X'_1$ and $\nact(X_2)$. By applying the inductive hypothesis we have that
$X'_1\fwd{\alpha\colorkey{i}} X'_1$ and we can derive $X'_1 + X_2 \fwd{\alpha\colorkey{i}} X_1 + X_2$ as desired. \qedhere
	 \end{description}
\end{proof}

Another fundamental property of causal-consistent reversibility is the so-called
causal-consitency~\cite{rccs,LanesePU20}, which essentially states
that we store the correct amount of causal information. 
In order to discuss it, we now borrow some definitions
from \cite{rccs}. We use $t, t', s, s'$ to range over transitions. In a transition $t:
X \red{\pi\colorkey{i}} Y$ we call $X$ the \emph{source} of the
transition, and $Y$ the \emph{target} of the transition. Two
transitions are said to be
\emph{coinitial} if they have the same source, and \emph{cofinal} if they have the same target. Given a transition
$t$, we indicate with $\underline{t}$ its reverse, that is if $t: X\fwd{\pi\colorkey{i}}Y$ (resp., $t: X\bk{\pi\colorkey{i}}Y$)
then $\underline{t}: Y\bk{\pi\colorkey{i}}X$ (resp., $\underline{t}: Y\fwd{\pi\colorkey{i}}X$). 
The notions of source, target, coinitiality, and cofinality naturally extend to paths.
We let $\chi,\omega$ to range
over sequences of transitions, which we call \emph{paths}, and with $\void_X$ we indicate the empty sequence starting and ending at $X$. We denote as $|\chi|$ the number of transitions in path $\chi$. 
Moreover, we indicate with $\chi_1\chi_2$ the composition of the two paths $\chi_1$ and $\chi_2$ when they are composable, that is when the target of $\chi_1$ coincides with the source of $\chi_2$.

\begin{defi}[Causal Equivalence]  Let $\ceq$ be the smallest equivalence on paths closed under composition and satisfying:
\begin{enumerate}
	\item if $t: X \red{\pi_1\colorkey{i}} Y_1$ and $s: X \red{\pi_2\colorkey{j}} Y_2$ are independent, and
	$s': Y_1\red{\pi_2\colorkey{j}} Z$, $t': Y_2\red{\pi_1\colorkey{i}} Z$
        then
	$ts' \ceq st'$;
	\item $t\underline{t} \ceq \void$ and $\underline{t}t \ceq \void$
\end{enumerate}
\end{defi} 
Intuitively, paths are causal equivalent if they differ only for
swapping independent transitions (we will discuss independence below)
and for adding do-undo or undo-redo pairs of transitions.

\begin{defi}[Causal Consistency (CC)] An LTS is causal consistent if for any coinitial and cofinal paths $\chi$ and $\omega$ we have
$\chi \ceq \omega$.
\end{defi}
Intuitively, if coinitial paths are cofinal then they have the same
causal information and can reverse in the same ways: we want only
causal equivalent paths to reverse in the same ways.

\subsection{Independence}
\label{sec:indepe}

We now define a notion of independence between $\rtpl$ coinitial transitions, based on a causality preorder (inspired by~\cite{LaneseP21}) on keys. Intuitively, independent transitions can be executed in any order (we will formalise this as Property~\ref{prp:square}), while transitions which are not independent represent a choice: either one is executed, or the other.

\begin{defi}[Partial order on keys]
The function $\ord(\cdot):\procs \mapsto 2^{(\keys \times \keys)}$ is inductively defined below.
It takes a configuration $X \in \procs$ and computes a set of ordered pairs of keys which is the set of causal relations among the keys in $X$.
\small
\[
\begin{array}{l}
\ord(P)=\emptyset   \qquad \qquad \ord(\res{X}{a}) = \ord(X)\\
 \ord(X\parallel Y) = \ord(X + Y) = \ord(\timed{X}{Y}) = \ord(X) \cup \ord(Y) \\
\ord(\rho\colorkey{i}.X)  = \ord(\timedl{X}{Y}{i}) = \ord(\timedr{Y}{X}{i}) = \{i < j \mid j\in \key(X)\} \cup \ord(X)
\end{array}
\]\normalsize
The partial order $\leq_{X}$ on $\key(X)$ is the reflexive and transitive closure of $\ord(X)$.
\end{defi}
Let us note that function $\ord$ computes a partial order relation, namely a set of pairs $(i,j)$, denoted $i < j$ to stress that they form a partial order. In particular, $i < j$ means that key $i$ causes key $j$. This takes into account both structural causality given by the structure of a configuration (e.g., a prefix causes its continuation) and causality raising from synchronisation and time, since synchronising actions and time actions corresponding to the same point in time have the same key.

\begin{example}
Let us compute the partial order on keys in  
$$\timedr{a}{b\colorkey{j}.P}{i} \parallel \ghost{i}.c\colorkey{k}.d\colorkey{w}.Q \parallel \ghost{i}.\co{c}\colorkey{k}.R$$ 
We have:
\[
\begin{array}{llll}
\ord(\timedr{a}{b\colorkey{j}.P}{i}) & = \{i < j\} \cup \ord(P) = \{i < j\} \cup \emptyset = \{i < j\} \\
\ord(\ghost{i}.c\colorkey{k}.d\colorkey{w}.Q) & = \{i < k \} \cup \{i < w \}  \cup \ord(c\colorkey{k}.d\colorkey{w}.Q)    \\
& =  \{i < k \} \cup \{i < w \}  \cup \{k < w\} \cup \ord(d\colorkey{w}.Q) \\
&= \{i < k \} \cup \{i < w \}  \cup \{k < w\} \cup \emptyset \cup \ord(Q) \\
&= \{i < k \} \cup \{i < w \}  \cup \{k < w\}\\
\ord(\ghost{i}.\co{c}\colorkey{k}.R) & =  \{i < k\} \cup \ord(R) = \{i < k\}
\end{array}
\]
and hence, looking at the parallel composition:
\begin{align*}  
&\ord(\timedr{a}{b\colorkey{j}.P}{i} \parallel \ghost{i}.c\colorkey{k}.d\colorkey{w}.Q \parallel \ghost{i}.\co{c}\colorkey{k}.R) = \\
&\{i<j\}\cup \{i < k \} \cup \{i < w \}  \cup \{k < w\} \cup \{i < k\}
=  \{i < j,\, i < k, \, i< w, \, k < w\} \tag*{$\diamond$}
\end{align*}
\end{example}

We also need to understand whether two forward communication transitions are in conflict since either they involve a same prefix or they involve different branches of a choice.

\begin{defi}[Forward communication conflict]
Two forward communication transitions $t_1: X \fwdns{\alpha_1\colorkey{i}} Y$ and $t_2: X \fwdns{\alpha_2\colorkey{j}} Z$ with $i \neq j$ are in forward communication conflict iff the $\fcc(Y,Z)$ predicate defined below holds:
\[
\begin{array}{lcl}
\fcc(\alpha\colorkey{i}.P,\alpha\colorkey{j}.P) & = & True\\
\fcc(P,P) & = & False\\
\fcc(Y_1 \parallel Y_2, Z_1 \parallel Z_2) & = & \fcc(Y_1,Z_2) \lor \fcc(Y_2,Z_2)\\
\fcc(Y_1 + Y_2,Z_1 + Z_2) & = &  (Y_1 \neq Z_1 \land Y_2 \neq Z_2) \lor \fcc(Y_1,Z_1) \lor \fcc(Y_2,Z_2) \\
\fcc(\res{Y_1}{a},\res{Z_1}{a}) & = & \fcc(Y_1,Z_1)\\
\fcc(\rho\colorkey{i}.Y_1,\rho\colorkey{i}.Z_1) & = & \fcc(Y_1,Z_1)\\
\fcc(\timedl{Y_1}{Y_2}{i},\timedl{Z_1}{Z_2}{j}) & = & \fcc(Y_1,Z_1)\\
\fcc(\timedl{Y_1}{Y_2}{i},\timedl{Z_1}{Z_2}{i}) & = & \fcc(Y_1,Z_1)\\
\fcc(\timedr{Y_1}{Y_2}{i},\timedr{Z_1}{Z_2}{i}) & = & \fcc(Y_2,Z_2)\\
\end{array}
\]
\end{defi}
For simplicity, the $\fcc$ predicate above is defined only for pairs of configurations which may arise from the same configuration.
Notice that in the clause for choice, the only way for the two
branches to be pairwise different, is that $Y$ has chosen one of them,
and $Z$ the other. In this case the two actions are in conflict. However, in this case $\fcc$ may not be defined on the components. To avoid this issue,
we consider the $\lor$ operator to be a short circuit operator.
\begin{example}
Let us consider the configuration $X_1 = a\colorkey{i}.b.P$ and the two transitions
\begin{itemize}
\item $t_1: X_1 \fwd{b\colorkey{j}}a\colorkey{i}.b\colorkey{j}.P = Y_1$ and 
\item $t_2: X_1 \fwd{b\colorkey{z}}a\colorkey{i}.b\colorkey{z}.P = Z_1$.
\end{itemize}
We have that 
\begin{align*}
\fcc(a\colorkey{i}.b\colorkey{j}.P, a\colorkey{i}.b\colorkey{z}.P) = 
\fcc(b\colorkey{j}.P,b\colorkey{z}.P) = True
\end{align*}
Let us consider the configuration $X_2 = \timedr{a.\nil}{b\colorkey{j}.(a.P + b.Q)}{i}$ and the two transitions
\begin{itemize}
\item $t_3: X_2 \fwd{a\colorkey{z}}\timedr{a.\nil}{b\colorkey{j}.(a\colorkey{z}.P + b.Q)}{i} = Y_2$ and 
\item $t_4: X_2 \fwd{a\colorkey{w}}\timedr{a.\nil}{b\colorkey{j}.(a.P + b\colorkey{w}.Q)}{i} = Z_2$.
\end{itemize}
We have that 
\begin{align*}
&\fcc\big(\timedr{a.\nil}{b\colorkey{j}.(a\colorkey{z}.P + b.Q)}{i}, \timedr{a.\nil}{b\colorkey{j}.(a.P + b\colorkey{w}.Q)}{i}\big) =\\
&\fcc\big(b\colorkey{j}.(a\colorkey{z}.P + b.Q), b\colorkey{j}.(a.P + b\colorkey{w}.Q)\big) = 
\fcc\big(a\colorkey{z}.P + b.Q, a.P + b\colorkey{w}.Q\big) = True
\end{align*} \finex
\end{example}
\begin{lem}
Function $\fcc$ above is total for each $Y$ and $Z$ obtained via
communication actions from a common $X$.
\end{lem}
\begin{proof}
We proceed by structural induction on $X$, with a case analysis on the rules used to derive the two transitions.
\begin{description}
\item[$X=\pi.P$] the only possibility here is that the prefix is executed, with two different keys, this case is covered by the first clause;
\item[$X=\timed{P}{Q}$] here the only possibility is that the first component is executed, this case is covered by the first clause for timeout;
\item[$X=X_1+X_2$] this case is covered by the fourth clause;
\item[$X=X_1 \parallel X_2$] this case is covered by the third clause;
\item[$X=\res{X_1}{a}$] this case is covered by the fifth clause;
\item[$X=A$] constant $A$ has a definition $A\stackrel{def}{=}P$, hence the proof for $P$ applies. Note that in this case termination by structural induction is not granted, but termination is ensured anyway since recursion is guarded;
\item[$X=\nil$] since $X$ cannot take any communication action, this case never applies;
\item[$X=\rho\colorkey{i}.X_1$] this case is covered by the sixth clause;
\item[$X=\timedl{X_1}{X_2}{i}$] this case is covered by the one but last clause;
\item[$X=\timedr{X_1}{X_2}{i}$] this case is covered by the last clause. \qedhere
\end{description}
\end{proof}

We now define a notion of conflict, and independence as its negation.

\begin{defi}[Conflict and independence]\label{def:confl}\label{def:ind}
Given a reachable $\rtpl$ configuration $\conf{T}{X}$,
two coinititial transitions $t: X \redns{\pi_1\colorkey{i}} Y$ and $s: X \redns{\pi_2\colorkey{j}} Z$ are conflicting, 
if and only if one of the following conditions holds:
\begin{enumerate}
        \item\label{confl:alfasigma} $t: X\fwdns{ \sigma\colorkey{i}} Y$ and $s: X\fwdns{ \alpha\colorkey{j}}Z$ or vice versa;
        \item\label{confl:consumed} $t: X\fwdns{ \pi_1\colorkey{i}} Y$ and $s: X\fwdns{ \pi_2\colorkey{j}}Z$ are in forward communication conflict;
			\item\label{confl:bkfwd} $t: X\fwdns{\pi_1\colorkey{i}}Y$ and $s: X\bkns{\pi_2\colorkey{j}}Z$ with $\ckey{j} \leq_{Y} \ckey{i}$ or vice versa;
			\item\label{confl:key} $t: X\fwdns{ \pi_1\colorkey{i}} Y$ and $s: X\fwdns{ \pi_2\colorkey{j}}Z$ with $\ckey{i}=\ckey{j}$.
\end{enumerate}
Transitions $t$ and $s$ are \emph{independent}, written $\conc{t}{s}$, if 
they are not conflicting. 
\end{defi}
Note that the conflict relation is reflexive and symmetric, hence independence is irreflexive and symmetric.
The first clause of Definition \ref{def:ind} tells us that a delay cannot be swapped with a communication action. Consider configuration $\timed{b.\nil}{\nil}$:  
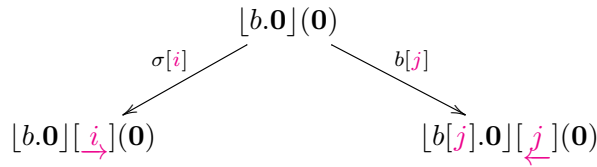
\begin{align*}
\xymatrix{
&&  \timed{b.\nil}{\nil} \ar[dl]_{\sigma\colorkey{i}}  \ar[dr]^{ b\colorkey{j}}&\\
 &\timedr{b.\nil}{\nil}{i} 	&& \timedl{b\colorkey{j}.\nil}{\nil}{j}
	& }
\end{align*}
Transitions $\sigma\colorkey{i}$ and $b\colorkey{j}$ are in conflict: they cannot be swapped since action $b$ is no longer possible after action $\sigma$, and vice versa. 

The second case of Definition~ \ref{def:ind} forbids either a same prefix or prefixes in different branches of a same choice operator to be consumed by the two transitions.
For example let us consider the configuration $a.\nil\parallel \co{a}.\nil$. The left configuration could execute an action $a\colorkey{i}$ while the entire configuration could synchronise by doing a $\tau\colorkey{j}$, as depicted below:
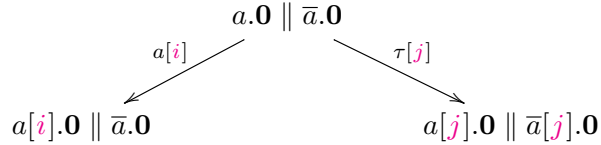
\begin{align*}
\xymatrix{
&&  a.\nil \parallel \co a.\nil \ar[dl]_{a\colorkey{i}}   \ar[dr]^{\tau\colorkey{j}} \\
 &a\colorkey{i}.\nil \parallel \co a.\nil && a\colorkey{j}.\nil \parallel \co a\colorkey{j}.\nil
}
\end{align*}
It is clear, from the example above, that the two actions cannot commute.

Another example of conflicting transitions captured by case~\ref{confl:consumed} of Definition~\ref{def:confl} is when transitions consume prefixes in different branches of a same choice operator. For example, let us consider the configuration $a.\nil+b.\nil$. The left branch can do an action $a\colorkey{i}$ while the right one an action $b\colorkey{j}$, as follows:
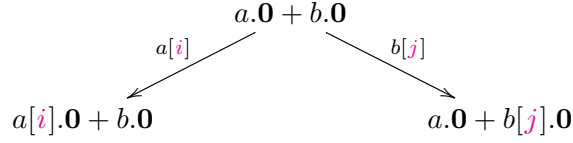
\begin{align*}
\xymatrix{
&&  a.\nil + b.\nil \ar[dl]_{a\colorkey{i}}   \ar[dr]^{b\colorkey{j}} \\
 &a\colorkey{i}.\nil + b.\nil && a.\nil+ b\colorkey{j}.\nil
}
\end{align*}
and again it is clear that these two transitions cannot commute. 

The third clause of Definition~\ref{def:ind} dictates that two transitions are in conflict when a reverse step eliminates some causes of a forward step. For example, the configuration $a\colorkey{j}.b.\nil$ can do a forward step with label $b\colorkey{i}$ going to $a\colorkey{j}.b\colorkey{i}.\nil$ or a backward one with label $a\colorkey{j}$, as follows: 
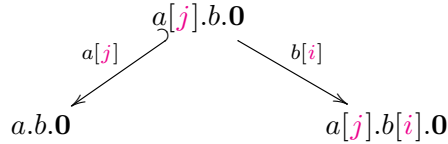
\begin{align*}
\xymatrix{
&&  a\colorkey{j}.b.\nil  \ar@{_{(}->}  [dl]_{a\colorkey{j}}\ar[dr]^{b\colorkey{i}} \\
 &a.b.\nil && a\colorkey{j}.b\colorkey{i}.\nil
}
\end{align*}
We have that $\ord(a\colorkey{j}.b\colorkey{i}.\nil)=\{j<i\}$, hence the side condition is satisfied. 
Undoing $a\colorkey{j}$ disables the action on $b$. 

The last case of Definition \ref{def:ind} forbids two transitions to pick  up the same key.

Notably, backward transitions are never in conflict, yet it is never
the case that a backward time action and a backward communication
action are enabled together, as shown by the following proposition.

\begin{prop}\label{prop:backalphasigma}
Let $X$ be a reachable $\rtpl$ configuration. Then it is never the case that $X \bk{\sigma\colorkey{i}} X'$ and $X \bk{\alpha\colorkey{j}} X''$.
\end{prop}
\begin{proof}
The proof is by structural induction on $X$. If $X$ is standard there is nothing to prove. If $X$ is a prefix $\rho\colorkey{k}.Y$ and $Y$ is standard then only communication actions are possible if $\rho$ is a communication action, only time actions otherwise. If $Y$ is not standard then the thesis follows by inductive hypothesis, since only rule [\textsc{Act}] is applicable. In the case of timeout, the thesis follows by noticing that at most on rule is applicable. In particular, for all the rules the thesis follows by inductive hypothesis but for rule [\textsc{STout}], for which it follows directly. The other cases follow by inductive hypothesis.  
\end{proof}

The Square Property tells that two coinitial independent transitions commute, thus closing a diamond. Formally:
\begin{restatable}[Square Property - SP]{Property}{squareproperty}
\label{prp:square}
Given a reachable $\rtpl$ configuration $\conf{T}{X}$ and 
 two coinititial transitions $t: X \redns{\pi_1\colorkey{i}} Y$ and $s: X \redns{\pi_2\colorkey{j}} Z$ with 
$\conc{t}{s}$ there exist two cofinal transitions $t':Y \redns{\pi_2\colorkey{j}} W$ and
$s': Z \redns{\pi_1\colorkey{i}} W$.
\end{restatable}
\begin{proof}
Deferred to Appendix~\ref{sec:proofs}.
\end{proof}

Since both CCSK and TPL are sub-calculi of $\rtpl$ as discussed in \cref{sec:corr}, then the notions of conflict and independence above induce analogous notions on CCSK and TPL. To the best of our knowledge, no such notion exists for TPL. Notions of conflict and independence (dubbed concurrency) for CCSK have been recently proposed in~\cite{Aubert22}, but they rely on extended labels while we define them on standard ones.

\subsection{Causal consistency}
We can now prove causal consistency, using the theory in~\cite{LanesePU20}.
The theory in ~\cite{LanesePU20} ensures that causal consistency follows from SP, already discussed, and two other properties: BTI (Backward Transitions are Independent) and WF (Well-Foundedness).
BTI generalises the concept of backward determinism used for reversible sequential languages~\cite{Janus}. It specifies that two backward transitions from a same configuration are always independent. 

\begin{restatable}[Backward Transitions are Independent - BTI]{Property}{bti}\label{prp:bti}
Given a reachable $\rtpl$ configuration $\conf{T}{X}$, 
any two distinct coinitial backward transitions $t: X \bkns{\pi_1\colorkey{i}} Y$ and $s: X \bkns{\pi_2\colorkey{j}} Z$ are independent.
\end{restatable}

BTI property trivially holds since (as mentioned above) by looking at the definition of conflicting and independent transitions (Definition~\ref{def:ind}) there are no cases in which two backward transitions are deemed as conflicting, hence
two backward transitions are always independent.

We now show that reachable configurations have a finite past.

\begin{restatable}[Well-Foundedness - WF]{Property}{wf}\label{prp:wf}
Let  $\conf{T_0}{X_0}$ be a reachable $\rtpl$ configuration. Then
 there is no infinite sequence such that $\conf{T_{i}}{X_{i}} \bkns{\pi_{i}\colorkey{j_i}} \conf{T_{i+1}}{X_{i+1}}$
 for all $i=0,1,\ldots$.
 \end{restatable}
\begin{proof}
WF follows since each backward transition removes
a key.
Given that the number 
$|\key(X)|$ of keys in $X$ is finite, only a finite amount of backward steps can be taken.
\end{proof}

The Parabolic Lemma~\cite[Lemma 11]{rccs}, stated below, tells us that any path is causally equivalent to a path made by only backward steps, followed by only forward steps. In other words, up to causal equivalence, paths can be rearranged so as
to first reach the maximum freedom of choice, going only backwards, and then continuing only
forwards.

\begin{defi}[Parabolic Lemma (PL)~{\cite[Lemma 11]{rccs}} property]
\label{lem:pl}
An LTS satisfies the Parabolic Lemma iff for any path $\chi$, there exist two forward-only paths $\omega,\omega'$ such that
$\chi \ceq \underline{\omega}\omega'$ and $|\omega| + |\omega'| \leq |\chi|$.
\end{defi}

We can now prove our main results thanks to the proof schema of~\cite{LanesePU20}.

\begin{prop}[cf. Proposition 3.4~\cite{LanesePU20}] Suppose BTI and SP hold, then PL holds.
\end{prop}

\begin{prop}[cf. Proposition 3.6~\cite{LanesePU20}] Suppose WF and PL hold, then CC holds.
\end{prop}
As a corollary of PL, reachable states are reachable via forward-only paths (cf.~\cite{LanesePU20}).
\begin{cor}\label{cor:fwreach}
A configuration $\conf{T}{X}$ is reachable 
iff there exists a process $\conf{0}{P}$ and a forward-only path $\conf{0}{P} \fwd{}^* \conf{T}{X}$.
\end{cor}
\begin{proof}
From PL, by noticing that the backward path is empty since $P$ cannot take backward actions.
\end{proof}
The general theory above can help us in proving specific properties of $\rtpl$, as we show below.

We have considered in this paper a global notion of time, as shown by the following theorem.
\begin{thm}\label{th:total}
For each reachable $\rtpl$ configuration $X$, the restriction of $\leq_{X}$ to keys attached to time actions is a total order. 
\end{thm}
\begin{proof}
From Corollary~\ref{cor:fwreach} we have that there exist a
process $P$ and a forward-only path $P \fwd{}^* X$. Take two arbitrary keys
$\ckey{i}$ and $\ckey{j}$ attached to time actions. Let $\ckey{i}$ be the first one to occur
in $P \fwd{}^* X$, and $X_1 \fwd{\sigma\colorkey{i}} X_2$ the transition introducing it (note that each step introduces a key). Since the path is forward, key $\ckey{j}$ will be attached to some configuration which is standard in $X_2$ (or to a $\ghostlab$ just before a standard configuration). We show by induction on the derivation of $X_1 \fwd{\sigma\colorkey{i}} X_2$ that this implies $\ckey{i} < \ckey{j}$. The thesis will follow. We have a case analysis on the last applied rule.

The cases of rules [\textsc{PAct}], [\textsc{RAct}] and [\textsc{Idle}] follow from the definition of $\ord$ on prefix. The cases of rules [\textsc{Act}], [\textsc{Hide}] and [\textsc{Const}] follow by induction. The cases for timeout are similar, noticing that $\ckey{j}$ could only be attached to the selected configuration. For parallel composition, only rule [\textsc{SynW}] needs to be considered, and the thesis follows by inductive hypothesis. Similarly, for choice only rule [\textsc{ChoW}] needs to be considered, and the thesis follows by induction as well. 
\end{proof}
As shown above, time actions are never independent, and only communication actions can be. Also, since time actions do not commute with communication actions (cf.~clause \ref{confl:alfasigma} in Definition~\ref{def:confl}) then each communication action is bound to be executed between two fixed time actions.

One may wonder whether the global notion of time described above is
too strict. This is a very good question, and indeed we plan in future
work to investigate different notions of causality for TPL, which will
induce a different causal-consistent reversible extension.

We show here just that dropping the $\ghost{i}$ prefix, which ensures time actions are recorded also by untimed configurations, would not solve the issue.
We have pursued this approach in~\cite{BocchiLMY22}, but it leads to violations of the Loop Lemma and the Parabolic Lemma, two main properties in the causal-consistency theory, as shown by the following example.

\begin{example}\label{ex:nocc}
Let us consider the configuration 
$X = \sigma.a.\nil \parallel b. \sigma.\nil$ and the following execution:
\begin{align*}
X \fwd{\sigma\colorkey{i}}\sigma\colorkey{i}.a.\nil \parallel b. \sigma.\nil \fwd{b\colorkey{j}}
\sigma\colorkey{i}.a.\nil \parallel b\colorkey{j}. \sigma.\nil\fwd{\sigma\colorkey{k}}
\sigma\colorkey{i}.a.\nil \parallel b\colorkey{j}. \sigma\colorkey{k}.\nil = Z
\end{align*}
Now from $Z$ we can undo the time actions $\sigma\colorkey{i}$ and $\sigma\colorkey{k}$
as follows:
$$Z \bk{\sigma\colorkey{i}} \sigma.a.\nil \parallel b\colorkey{j}. \sigma\colorkey{k}.\nil = Z_1 \bk{\sigma\colorkey{k}} \sigma.a.\nil \parallel b\colorkey{j}. \sigma.\nil =Z_2$$
Now let us focus on the last transition. According to the Loop Lemma (Lemma~\ref{lem:loop}) we can reach $Z_1$ from $Z_2$ by doing a forward time action, that is
 $Z_2\fwd{\sigma\colorkey{k}}Z_1$, but this is impossible as
  $$\sigma.a.\nil \parallel b\colorkey{j}. \sigma.\nil \fwd{\sigma\colorkey{k}} \sigma\colorkey{k}.a.\nil \parallel b\colorkey{j}. \sigma\colorkey{k}.\nil \neq Z_1$$
 Also, the Parabolic Lemma fails. Indeed if we consider the path which leads to $Z_1$, according to 
 the Parabolic Lemma, we can rewrite this path as a sequence of backward transitions followed by forward ones.  If from $Z_1$ we undo all the actions and try to reach it by using just forward actions we fail since:
 \begin{align*}
 Z_1= \,& \sigma.a.\nil \parallel b\colorkey{j}. \sigma\colorkey{k}.\nil \bk{\sigma\colorkey{k}}
 \bk{b\colorkey{j}} \sigma.a.\nil \parallel b.\nil\\
 & \fwd{b\colorkey{j}} \sigma.a.\nil \parallel b\colorkey{j}. \sigma.\nil \fwd{\sigma\colorkey{k}} \sigma\colorkey{k}.a.\nil \parallel b\colorkey{j}. \sigma\colorkey{k}.\nil \neq Z_1
 \end{align*}
 
By using $\ghost{i}$ prefixes we impose a total order among time actions, as shown in Theorem~\ref{th:total}, as follows:
\begin{align*}
X \fwd{\sigma\colorkey{i}}&\sigma\colorkey{i}.a.\nil \parallel \ghost{i}.b. \sigma.\nil \fwd{b\colorkey{j}}
\sigma\colorkey{i}.a.\nil \parallel \ghost{i}.b\colorkey{j}. \sigma.\nil\fwd{\sigma\colorkey{k}} \\
&\sigma\colorkey{i}.\ghost{k}.a.\nil \parallel b\colorkey{j}. \sigma\colorkey{k}.\nil = Y
\end{align*}
Now from $Y$ we cannot undo the time action $\sigma\colorkey{i}$, since now we need to undo action $\sigma\colorkey{k}$ first. With this machinery in place, we enforce a strict notion of causality in TPL, but we have been able to successfully build a causal-consistent reversible extension.\finex
\end{example}

%

\section{Conclusion, Related and Future Work}\label{sec:conc}

The main contribution of this paper is the study of the interplay between causal-consistent reversibility and time. 
A reversible semantics for TPL cannot be automatically derived using well-established frameworks~\cite{ccsk,LaneseM20}, since some operator acts differently depending on whether the label is a communication or a time action. For example, in TPL a choice cannot be decided by the passage of time, making the $+$ operator both static and dynamic, and the approach in~\cite{ccsk} not applicable. To faithfully capture patient actions in a reversible semantics we introduced $\ghostlab$ prefixes.
Another peculiarity of TPL is the timeout operator $\timed{P}{Q}$, which can be seen as a choice operator whose left branch has priority over the right one. 
Indeed, if $P$ can do a $\tau$ action then $Q$ can not execute and it is discarded. 
Although we have been able to use the static approach to reversibility~\cite{ccsk}, adapting it to our setting has been challenging for the aforementioned reasons.
Notably, our results have a double interpretation: as an extension of CCSK~\cite{ccsk} with time, and as a reversible extension of TPL~\cite{tpl}. As a side result, by focusing on the two fragments, we derive notions of independence and conflict for CCSK and TPL.

\paragraph{Other process algebras}
As said by Baeten and Bergstra
\emph{``Adding real time features can be done in many ways and it is impossible to explore all options in a single paper''}~\cite{2}. The literature of timed process calculi is indeed rich. Thus, we only give an outline of the main approaches with the purpose of reflecting on the applicability of our results to different time approaches.
Besides TPL~\cite{tpl}, considered in this paper, a non-exhaustive list of alternative formalisms includes timed CSP~\cite{37}, temporal CCS~\cite{30}, timed CCS~\cite{Yi91}, real-time ACP~\cite{2}, urgent LOTOS~\cite{9}, CIPA~\cite{1}, ATP~\cite{31}, TIC~\cite{36}, PAFAS~\cite{pafas}, and mCRL2~\cite{mcrl}. 

To simplify the discussion, we build on the categorisation in~\cite{BM2023a} and focus our comparison on the following time-related design choices:
\begin{description}
\item[Separated vs integrated semantics] In the first case, actions are
instantaneous and time only passes in between actions; hence, functional
behaviour and time are orthogonal. In the second case, every action takes a
certain amount of time to be performed and time passes only due to action
execution; hence, functional behaviour and time are integrated.
\item[Relative time vs absolute delays] In the first case, each
delay refers to the time instant of the previous observation. In the
second case, all delays refer to the starting time of the system's execution.
\item[Global clock vs local clocks] In the first case, a single clock governs the pace of time
passing in the system. In the second case, several clocks associated with the various system
parts may have different views of the pace of time. If a model allows processes to have local clocks but time flows at the same pace for all of them (even if they hold different values due to resets, as in the case e.g. of Timed Automata~\cite{TA}) we still classify the model as a global clock model. 
\item[Eager vs lazy vs maximal progress] There are several interpretations of when a communication action can be executed or delayed. Eager semantics enforce actions to be performed as soon as they become enabled, i.e.,
without any delay, thereby implying that their execution is urgent. On the other hand, laziness allows the execution of an action to be delayed even if the action is enabled. Maximal progress is eager for internal actions and lazy otherwise: actions can be delayed only if they are waiting to synchronise with some external partner which is not yet available.
Some calculi have primitives for both eager and lazy actions, so each action can be either lazy or eager. 
\end{description}
Table~\ref{tablesurvey} illustrates how the aforementioned timed calculi position with respect to the four criteria above. 
Most of the formalisms we have reviewed combine \textbf{separated semantics}, \textbf{relative time}, and \textbf{global clock}. The main difference is the urgency (or lack thereof) of communication actions with respect to time actions. 
ATP~\cite{31}, temporal CCS~\cite{30}, and PAFAS~\cite{pafas} allow actions to happen at any time within the prescribed intervals (e.g., later than when they become ready to execute). Instead, timed CSP~\cite{37} and timed CCS~\cite{Yi91} share the same approach we adopted in this paper, inherited from TPL~\cite{tpl}: actions are normally lazy, unless they are silent in which case they are eager (maximal progress). A more general approach is the one of urgent LOTOS~\cite{9}, which provides primitives for urgent actions and primitives for non-urgent actions, hence enabling one to decide the semantics of each specific action. 
The remaining formalisms have \textbf{integrated semantics} combined with \textbf{absolute time}. 
In mCRL~\cite{mcrl}, CIPA~\cite{36}, and TIC~\cite{36}, the transition relation models both execution of actions and time elapsing (integrated semantics).\footnote{The transition relation of mCRL does also feature an idling relation, but this does not lead to any follow-up state and is just for final states.} In all the three calculi, time is specified from the beginning of the computation (absolute time). While mCRL relies on a global clock, CIPA and TIC allow parallel processes to go `out of sync' (local clocks). In CIPA, global time synchronisation is only required for causally dependent actions (it has to be re-established before two processes can communicate with each other). TIC uses an `age' function to record discrepancies between the time of parallel processes. mCRL has no silent actions, and time idling and communication actions can happen at any time, after they become ready (lazy). In CIPA, the timing of an action needs to exactly match its prescription so the action happens as soon as it is ready (eager). TIC allows delays of exact amounts of time (urgent/eager) as well as delays of times within an interval (lazy).

The application of our approach using integrated semantics and/or absolute time should not present any particular challenge. In fact, separated and integrated semantics have been shown to be equivalent~\cite{BernardoCT16} (i.e., they can be encoded into each other preserving weak barbed bisimilarity). Similarly for absolute instead of relative time thanks to the equivalence given in~\cite{Corradini00}. 

Extending our framework to local clocks (e.g., as in CIPA) would be interesting but non-trivial in our integrated semantics. It may require us to record some live information on the different time perspective of parallel processes to rule out unwanted interleavings. An alternative could be to exploit the encoding of \cite{BernardoCT16}
from TCCS to CIPA, and see whether the semantics is still preserved while considering reversible behaviours.

Building from the conference version of this article~\cite{BocchiLMY22}, the work in~\cite{BM2023a} has shown that our approach would apply also to a semantics with only eager actions and to a semantics with only lazy actions. However, the applicability of our approach to a scenario where each action can be statically set to be either lazy or eager (the `either' option in the `action' column of Table~\ref{tablesurvey}) needs to be further investigated.

\begin{table}[]
\begin{tabular}{l|l|l|l|l}
 &\textbf{semantics} & \textbf{time} & \textbf{clocks} & \textbf{actions} \\ 
 \hline \hline &&&&\\[-3mm]
 ATP~\cite{31} &  separated & relative  & global & lazy \\
 temporal CCS~\cite{30} &  separated & relative  & global & lazy \\
 PAFAS~\cite{pafas}& separated & relative & global & lazy \\
 TPL~\cite{tpl} & separated & relative & global &  maximal progress\\
 timed CSP~\cite{37} & separated & relative & global &  maximal progress  \\
timed CCS~\cite{Yi91} & separated & relative & global &  maximal progress \\
urgent LOTOS~\cite{9} & separated & relative & global & either\\
 \hline
 mCRL2~\cite{mcrl} & integrated & absolute & global & lazy \\
 CIPA~\cite{1} & integrated & absolute & local & eager\\
 TIC~\cite{36}& integrated & absolute & local & either\\
 \end{tabular}
\caption{Semantics can be separated or integrated; time can be relative or absolute; clocks can be global or local; actions can be eager, lazy, either of them, or maximal progress.}
\label{tablesurvey}
\end{table}

\paragraph{Alternative timed formalisms}
Timed Petri nets are a relevant tool for analysing real-time systems. A step towards the analysis of real-time systems would be to encode $\rtpl$
into timed Petri nets~\cite{ZimmermannFH01} extended with reversibility, by building on the encoding of reversible CCS into reversible Petri nets~\cite{MelgrattiMP21}. Also, we could think of encoding $\rtpl$ in timed automata~\cite{AlurCD93} extended with reversibility. 
Another possibility would be to study the extension of a monitored timed semantics for multiparty session types, as the one of~\cite{NeykovaBY17}, with reversibility~\cite{lmcs:8540}.

Maximal progress of TPL (as well as $\rtpl$) has connections with Markov chains~\cite{BrinksmaH00}. For instance,
in a stochastic process algebra, the process 
$$\tau.P + (\lambda).Q $$
 (where $\lambda$ is a rate) will not be delayed since  $\tau$ is
instantaneously enabled. This is similar to maximal progress for the timeout operator.
A deep comparison between deterministic time, used by TPL, and stochastic time, used by stochastic process algebras, can be found in~\cite{BernardoCT16}. Further investigation on the relation between our work and~\cite{BMezzina20}, studying reversibility in Markov chains, is left for future work. 
The treatment of passage of time shares some similarities with broadcast~\cite{Mezzina18} as well: time actions affect parallel components in the same way.

\paragraph{Future directions}
We have just started our research quest towards a reversible timed semantics. Beyond considering local notions of time, as discussed after Theorem~\ref{th:total}, a first improvement would be to add an explicit rollback operator, as in~\cite{LaneseMSS11}, that could be triggered, e.g., in reaction to a timeout. Also, asynchronous communications (like in Erlang) could be taken into account. TPL is a conservative timed extension of CCS. Due to its simplicity, it has a very clear behavioural theory~\cite{tpl}, including an axiomatization.  A further step could be to adapt such behavioural theory to account for reversibility, by combining it with the one for CCSK developed in~\cite{LaneseP21}. However, the fact that reversibility breaks Milner's expansion law may limit the power of the axiomatisation. Also, we could consider studying more complex temporal operators~\cite{NicollinS91}. In TPL time is discrete, and the language abstracts away from how time is represented. Indeed, the idling prefix $\sigma$ is meant to await one cycle of clock. A more fine-grained treatment of time in CCS was proposed in Timed CCS (TCCS)~\cite{Yi90,Yi91}. In TCCS it is possible to express a process, say P, which awaits 3 time units directly by: 
\[\epsilon(3).P\]
 Now  the process above, in principle, can be rendered in TPL as the process $\sigma.\sigma.\sigma.P$ by assuming that a cycle of clock lasts one time unit. But this is only possible if we consider TCCS with discrete time. Even if we restrict ourselves to discrete time, encoding the $\epsilon(\cdot)$ operator in TPL would be troublesome to treat (from a reversible point of view) as a single step has to be matched by several ones. Also, TCCS obeys to \emph{time additivity} (two actions taking times $t_1$ and $t_2$ can be turned into a single action taking time $t_1+t_2$), while TPL does not. 
 As shown in \cite{BM2023a}, time additivity poses a problem with our approach: in presence of time additivity, the proof schema proposed in ~\cite{LanesePU20} does not hold anymore. In particular,
 because of time additivity BTI does not hold anymore and Loop Lemma has to be 
 formalised in a weaker form. Hence, one has to redo all the proofs. 
 For all these reasons, it will not be straightforward to adapt the approach in this paper to deal with TCCS.

\paragraph{Prospective applications} As discussed above, this work is a first step towards an analysis of
reversible real-time systems and it has the purpose of clarifying the
relationship between reversibility and time.  Although the
contribution of this work is theoretical, we envisage a potential
application to debugging of real-time Erlang code.  More concretely, we would
like to extend CauDEr~\cite{LaneseNPV18,GV21-CauREr,CauDEr}, the only
causal-consistent reversible debugger for a (fragment of a) real
programming language we are aware of. The purpose of the extension would be to support \emph{timed} Erlang
programs. To this end we would first need to extend the reversible
semantics of Erlang in~\cite{LamiLSCF22,FabbrettiLS21} with a notion
of time, imported from the present work, so to support constructs such as `\lstinline{after}'
and `\lstinline{sleep}', as used, e.g., in our
Example~\ref{ex:1}. The `\lstinline|after|' (i.e., timeout) construct, in particular, is
very common in the Erlang programming practice.
Even if Erlang timeouts are close to TPL ones, there are a
number of challenges to be faced. First, Erlang communication is
asynchronous, unlike $\rtpl$. Second, and more importantly, Erlang
delays can be explicit in the code, as in our Example~\ref{ex:1}, but
they can also be generated by network delays or long computations.
Therefore, in order to enable reversible debugging
of timed programs, one needs to pair the code with a model, possibly computed in an automated way, that
describes the delays that are likely to occur in
the system of interest.
The development of this
prospective application goes beyond the scope of the formal setting
given in the current work.

Another possible application is to bring our theory to timed Rebecca \cite{timedRebecca} which is a timed actor based language with model checking support. This would enable us to exploit model checking for reversible behaviours.

\bibliographystyle{alphaurl}

\bibliography{biblio}
\FloatBarrier
\appendix
\section{Background: CCS, CCSK and TPL}\label{sec:bg}

In this section we present the full syntax and semantics of CCS, TPL and CCSK, taken respectively from
~\cite{ccs},~\cite{tpl} and~\cite{ccsk}.

\subsection{CCS: Calculus of Communicating Systems}\label{app:ccs}
\begin{figure}[t]
\begin{align*}
\textrm{(Processes)}\quad & P = \,   \alpha.P  \sep  P + Q \sep P \parallel Q \sep \res{P}{a} \sep A \sep \nil & 
\end{align*}
\caption{Syntax of CCS}
\label{fig:ccssyn}
\end{figure}
\begin{figure}[t]
\begin{mathpar}
\inferrule*[left=(Act)]{}{\alpha.P \tplfwd{\alpha}P} \and
\inferrule*[left=(Sum$_1$)]{P\tplfwd{\alpha} P'}{P+Q \tplfwd{\alpha} P' + Q} \and
\inferrule*[left=(Sum$_2$)]{Q\tplfwd{\alpha} Q'}{P+Q \tplfwd{\alpha} P + Q'} \and
\inferrule*[left=(Com$_1$)]{P\tplfwd{\alpha} P'}{P\parallel Q \tplfwd{\alpha} P' \parallel Q} \and
\inferrule*[left=(Com$_2$)]{Q\tplfwd{\alpha} Q'}{P\parallel Q\tplfwd{\alpha} P \parallel Q'} \and
\inferrule*[left=(Com$_3$)]{P\tplfwd{\alpha} P' \and Q\tplfwd{\co{\alpha}} Q'}{P\parallel Q\tplfwd{\tau} P' \parallel Q'} \\
\inferrule*[left=(Res)]{P\tplfwd{\alpha} P' \and \alpha\not\in \{a,\co{a}\}}{\res{P}{a}\tplfwd{\tau} \res{P'}{a}} \and
\inferrule*[left=(Rec)]{A\stackrel{def}{=}P \and P\tplfwd{\alpha} P' }{A\tplfwd{\alpha} P'} \end{mathpar}
\caption{Semantics of CCS}
\label{fig:ccssem}
\end{figure}
The Calculus of Communicating Systems is a process calculus introduced by Milner~\cite{ccs}.
We let $\names$ be the set of action names $a$, $\co{\names}$ the set of action conames $\co{a}$. We use $\alpha$ to range over $a$, $\co{a}$ and internal actions $\tau$. We assume $\co{\co{a}}=a$. 
We let $\mathcal{A}^{\tau}= \names \cup \co{\names} \cup \{\tau\}$.
The syntax of CCS is reported in Figure~\ref{fig:ccssyn}. A process can be an action \emph{prefix}
$\alpha.P$, that can perform an action $\alpha$ and continue as $P$, a non-deterministic \emph{choice}
$P+Q$ among two processes, a \emph{parallel} composition $P \parallel Q$ of two processes, the \emph{restriction}
$\res{P}{a}$, which acts as the process $P$, but actions on $a$ are forbidden, the constant \emph{identifier} $A$
and the \emph{inactive} process $\nil$. The semantics of CCS is given by the labelled transition system (LTS)
$(\mathcal{P}, \mathcal{A}^{\tau}, \tplfwd{})$, where $\mathcal{P}$ is the set of all CCS processes,
$\mathcal{A}^{\tau}$ is the set of labels and $ \tplfwd{}$ is the least transition relation induced by the rules in Figure~\ref{fig:ccssem}.

\subsection{CCSK: CCS with communication keys}\label{app:ccsk}
CCS with communication keys (CCSK) is a reversible extension of CCS obtained by applying to CCS the approach in~\cite{ccsk}.
The key idea of this approach is to make all the dynamic operators (such as prefix and non-deterministic choice) static
and to decorate prefixes with freshly created identifiers, dubbed communication keys, when they are executed.
The syntax of CCSK is reported in Figure~\ref{fig:ccsksyn}, where the addition with respect to the syntax of CCS (Figure~\ref{fig:ccssyn}) is enclosed in grey boxes. With respect to CCS, in CCSK a prefix $\alpha$ can be decorated with a communication key $\magenta{i}$, to indicate the fact that the prefix has been executed. As one can see,
CCSK (reversible) configurations are built on top of CCS processes.

\begin{figure}[t]
\begin{align*}
\textrm{(Processes)}\quad & P = \,   \alpha.P  \sep  P + Q \sep P \parallel Q \sep \res{P}{a} \sep A \sep \nil &   \\[0.2cm]
\textrm{(Configurations)}\quad  &  X = \colorbox{shade}{$\alpha\colorkey{i}.X$} \mid X + Y   \sep X \parallel Y \sep X\setminus a \sep P 
\end{align*}
\caption{Syntax of CCSK}
\label{fig:ccsksyn}
\end{figure}

\begin{figure}[t]
\begin{mathpar}
\inferrule*[left=(Act1)]{\std(X)}{\alpha.X \ccskfwd{\alpha\colorkey{i}} \alpha\colorkey{i}.X} \and
\inferrule*[left=(Act2)]{X\ccskfwd{\beta\colorkey{j}}X' \\ i\neq j}{\alpha\colorkey{i}.X \ccskfwd{\beta\colorkey{j}}\alpha\colorkey{i}.X'}\\
\inferrule*[left=(Sum1)]{X\ccskfwd{\alpha\colorkey{i}} X' \and \std(Y)}{X+Y \ccskfwd{\alpha\colorkey{i}} X' + Y} \and
\inferrule*[left=(Sum2)]{Y\ccskfwd{\alpha\colorkey{i}} Y' \and \std(Y)}{X+Y \ccskfwd{\alpha\colorkey{i}} X + Y'} \\
\inferrule*[left=(Par1)]{X\ccskfwd{\alpha\colorkey{i}} X' \and i\not \in \key(Y)}{X\parallel Y \ccskfwd{\colorkey{i}} X' \parallel Y} \and
\inferrule*[left=(Par2)]{Y\ccskfwd{\alpha\colorkey{i}} Y' \and i\not \in\key(X)}{X\parallel Y \ccskfwd{\alpha\colorkey{i}} X \parallel Y'}\and
\inferrule*[left=(Par3)]{X\ccskfwd{\alpha\colorkey{i}} X' \and Y\ccskfwd{\co{\alpha}\colorkey{i}} Y' }{X\parallel Y \ccskfwd{\tau\colorkey{i}} X \parallel Y'} \\
\inferrule*[left=(Res)]{X \ccskfwd{\alpha\colorkey{i}} X' \\ \alpha\not\in \{a,\co{a}\}}
{\res{X}{a}\ccskfwd{\alpha\colorkey{i}}\res{X'}{a}} \and
\inferrule*[left=(Rec)]{A\stackrel{def}{=}P \and P\ccskfwd{\alpha\colorkey{i}}X} {A\ccskfwd{\alpha\colorkey{i}}X}
\end{mathpar}
\caption{Forward LTS of CCSK}
\label{fig:ccskfwlts}
\end{figure}

\begin{figure}[t]
\begin{mathpar}
\inferrule*[left=(Act1)]{\std(X)}{\alpha\colorkey{i}.X \ccskbk{\alpha\colorkey{i}}a.X} \and
\inferrule*[left=(Act2)]{X \ccskbk{\alpha\colorkey{i}}X' \and i\neq j}{\alpha\colorkey{i}.X \ccskbk{\alpha\colorkey{i}}a.X} \\
\inferrule*[left=(Sum1)]{X\ccskbk{\alpha\colorkey{i}}X' \and \std(Y)}{X+Y \ccskbk{\alpha\colorkey{i}} X'+Y} \and
\inferrule*[left=(Sum2)]{Y\ccskbk{\alpha\colorkey{i}}Y' \and \std(X)}{X+Y \ccskbk{\alpha\colorkey{i}} X+Y'}\\
\inferrule*[left=(Par1)]{X\ccskbk{\alpha\colorkey{i}}X' \and i\not\in \key(Y)}{X\parallel Y \ccskbk{\alpha\colorkey{i}} X' \parallel Y} \and
\inferrule*[left=(Par2)]{Y\ccskbk{\alpha\colorkey{i}}Y' \and i\not\in \key(X)}{X\parallel Y \ccskbk{\alpha\colorkey{i}} X \parallel Y'} \\
\inferrule*[left=(Par3)]{X\ccskbk{\alpha\colorkey{i}} X' \and Y\ccskbk{\co{\alpha}\colorkey{i}} Y' }{X\parallel Y \ccskbk{\tau\colorkey{i}} X \parallel Y'} \\
\inferrule*[left=(Res)]{X \ccskbk{\alpha\colorkey{i}} X' \\ \alpha\not\in \{a,\co{a}\}}
{\res{X}{a}\ccskbk{\alpha\colorkey{i}}\res{X'}{a}} \and
\inferrule*[left=(Rec)]{A\stackrel{def}{=}P \and X\ccskbk{\alpha\colorkey{i}}P} {P\ccskfwd{\alpha\colorkey{i}}A}
\end{mathpar}
\caption{Backward LTS of CCSK}
\label{fig:ccskbklts}
\end{figure}
The semantics of CCSK is given by
two LTSs, the \emph{forward} one $(\mathcal{X}^{k}, \mathcal{A}^{\tau}  \times \keys, \ccskfwd{})$
and the \emph{backward} one $(\mathcal{X}^{k}, \mathcal{A}^{\tau}  \times \keys, \ccskbk{})$,
where $\mathcal{X}^{k}$ is the set of CCSK configurations, $\keys$ is the set of keys and
$\ccskfwd{}$ and $\ccskbk{}$ are the least transition relations induced by the rules in Figure~\ref{fig:ccskfwlts} and
Figure~\ref{fig:ccskbklts}, respectively. Since CCSK is reversibile,
for each forward rule there exists a corresponding backward one. The two LTSs use two functions, $\std(X)$ and
$\key(X)$. Intuitively, function $\std(X)$ states that a configuration $X$ has no decorated prefixes (i.e., it does not contain any history information hence it is a CCS process), while function $\key(X)$
returns the set of keys of a given configuration.

\subsection{TPL: Timed Process Language}\label{app:tpl}

\begin{figure}[ht]
\begin{align*}
\textrm{(Processes)}\quad & P = \,   \pi.P \sep \colorbox{shade}{$\timed{P}{Q}$} \sep  P + Q \sep P \parallel Q \sep \res{P}{a} \sep A \sep \nil & (\pi =\, \alpha \sep \colorbox{shade}{$\sigma$}  )  
\end{align*}
\caption{Syntax of TPL}
\label{fig:tplsyn}
\end{figure}

TPL \cite{tpl} is an extension of CCS with time. Its syntax is reported in
Figure~\ref{fig:tplsyn}, where the novelties w.r.t.~CCS are enclosed in grey boxes.
TPL adds to CCS two constructs to deal with time: the \emph{idling} prefix $\sigma.P$
and the \emph{timeout} operator $\timed{P}{Q}$. The process $\sigma.P$ acts
as $P$ after having waited one time unit, while the timeout $\timed{P}{Q}$ executes $P$ (if possible)
or $Q$ (in case of a timeout). We indicate with  $\mathcal{A}^{t}$ the set
 $\names \cup \co{\names} \cup \{\tau,\sigma\}$.
The semantics of TPL is given via an LTS and two set of rules:
rules for actions (Figure~\ref{fig:tplsem1}, similar to the CCS rules with the addition of rule
\textsc{Then}), and rules for time (Figure~\ref{fig:tplsem2}), which regulate the behaviour
of the temporal operators $\sigma.P$ and $\timed{P}{Q}$ and the passage of time in the system (e.g., in a parallel composition). Hence, the semantics of TPL is given by the LTS $(\mathcal{P}^t, \mathcal{A}^t, \tplfwd{})$, where
$\mathcal{P}^t$ is the set of TPL processes and $ \tplfwd{}$ is the least relation induced by the rules in
Figures~\ref{fig:tplsem1} and~\ref{fig:tplsem2}.
\begin{figure}[ht]
\begin{mathpar}
\inferrule*[left=(Act)]{}{\alpha.P \tplfwd{\alpha}P} \and
\inferrule*[left=(Sum$_1$)]{P\tplfwd{\alpha} P'}{P+Q \tplfwd{\alpha} P' + Q} \and
\inferrule*[left=(Sum$_2$)]{Q\tplfwd{\alpha} Q'}{P+Q \tplfwd{\alpha} P + Q'} \and
\inferrule*[left=(Then)]{P\tplfwd{\alpha} P'}{\timed{P}{Q}\tplfwd{\alpha} P'} \and
\inferrule*[left=(Com$_1$)]{P\tplfwd{\alpha} P'}{P\parallel Q \tplfwd{\alpha} P' \parallel Q} \and
\inferrule*[left=(Com$_2$)]{Q\tplfwd{\alpha} Q'}{P\parallel Q\tplfwd{\alpha} P \parallel Q'} \and
\inferrule*[left=(Com$_3$)]{P\tplfwd{\alpha} P' \and Q\tplfwd{\co{\alpha}} Q'}{P\parallel Q\tplfwd{\tau} P' \parallel Q'} \\
\inferrule*[left=(Res)]{P\tplfwd{\alpha} P' \and \alpha\not\in \{a,\co{a}\}}{\res{P}{a}\tplfwd{\tau} \res{P'}{a}} \and
\inferrule*[left=(Rec)]{A\stackrel{def}{=}P \and P\tplfwd{\alpha} P' }{A\tplfwd{\alpha} P'} \end{mathpar}
\caption{LTS of TPL: rules for actions}
\label{fig:tplsem1}
\end{figure}

\begin{figure}[ht]
\begin{mathpar}
\inferrule*[left=(Act$_2$)]{}{\alpha.P \tplfwd{\sigma}\alpha.P} \and
\inferrule*[left=(Nil)]{}{\nil \tplfwd{\sigma}\nil} \and
\inferrule*[left=(Wait)]{}{\sigma.P \tplfwd{\sigma}P}  \\
\inferrule*[left=(Sum$_3$)]{P \tplfwd{\sigma}P' \and Q \tplfwd{\sigma}Q'}{P+Q \tplfwd{\sigma} P' + Q'}  \and
\inferrule*[left=(Then$_2$)]{P\not\tplfwd{\tau}}{\timed{P}{Q} \tplfwd{\sigma}Q} \\
\inferrule*[left=(Com$_4$)]{P\tplfwd{\sigma} P' \and Q\tplfwd{\sigma}Q' \and  (P\parallel Q)\not\tplfwd{\tau}}
{P\parallel Q \tplfwd{\sigma}P'\parallel Q'} \and
\inferrule*[left=(Res$_2$)]{P\tplfwd{\sigma} P'}{\res{P}{a}\tplfwd{\sigma} \res{P'}{a}} \and
\inferrule*[left=(Rec$_2$)]{A\stackrel{def}{=}P \and P\tplfwd{\sigma} P' }{A\tplfwd{\sigma} P'} 
\end{mathpar}
\caption{LTS of TPL: rules for time}
\label{fig:tplsem2}
\end{figure}


\section{Encoding of negative premises}
\label{sec:negative}

In this section we show that there exists an encoding of the negative premises in the rules of Figure~\ref{fig:fw} and Figure~\ref{fig:bk} into decidable positive premises.  To do so we compute all the enabled forward prefixes (i.e., barbs) of a configuration and form all the possible pairs of prefixes on the two sides of a parallel operator. We then check whether there exists a pair containing both an action and the corresponding co-action.
The  operator $\synch$ is inductively defined as follows:

\begin{defi}[Synchronisation Operator]
\begin{align*}
& \synch(\alpha.P) =\{ \{ \alpha \}\} &  \synch(\rho\colorkey{i}.X) = \synch(X)\\
&	\synch(X\parallel Y) = \synch(X) \oplus \synch(Y)& \synch(\res{X}{a}) = \synch(X) \setminus \{a\}\\
  & \synch(X+Y) = \synch(X)\cup\synch(Y) \text{ if } \nact(X+Y)\\
  & \synch(X+Y) = \synch(X) \text{ if } \neg\nact(X)
  & \synch(X+Y) = \synch(Y) \text{ if } \neg\nact(Y)\\  
&\synch(\timed{X}{Y}) = \synch(X) & \synch(\timedl{X}{Y}{i}) = \synch(X) \\
& \synch(\timedr{X}{Y}{i}) = \synch(Y)
& \synch(A) = \synch(P) \text{ if } A \stackrel{def}{=}P\\
& \synch(\nil) = \emptyset & \synch(\sigma.P) = \emptyset
\end{align*}
where $A\oplus B$ 
is defined as follows:
$$\{A_i \sep i\in I\} \oplus \{B_j \sep j\in J\} = \bigcup_{i\in I, j\in J} \{A_i\cup B_j\} $$
We can then define $X\not \fwd{\tau}$ as:
$$X\not \fwd{\tau} = \forall C \in \synch(X). \forall c_i,c_j\in C \, .   \co{c_i} \neq c_j$$
\end{defi}
The intuition behind the $\synch(\cdot)$ operator is that for sequential processes (i.e., processes which have no top level $\parallel$) it computes a set of singletons of prefixes. Such a set represents the list of all the enabled forward prefixes, which could synchronise in a parallel composition. This is rendered by the rule $\synch(X \parallel Y) = \synch(X)\oplus \synch(Y)$. In this case, via the operator $\oplus$, we compute all the possible pairs of such singletons. Let us note that if we have more than one top level $\parallel$, say $n$, we will have a set of $(n+1)$-uples. Also, since we are using set-based operators, repetitions of prefixes will be dropped.

The $\synch$ operator just collects all the available prefixes, discarding the ones already executed, and the discarded branches.
 For example, let us consider the process $P = a+\co{a}$ we have that
$\synch(P) = \{\{a\},\{\co a\}\}$. Synchronisation is induced by the parallel operator $\parallel$, hence when the $\synch$ operator meets a parallel composition, it computes all the possible pairs via the $\oplus$ operator. We then have a possible synchronisation if there is a pair of the form $\{\alpha,\co{\alpha}\}$. For example, if we take the process $P$ above and we put it in parallel with $a$, the synchronisation operator will compute the following pairs 
$$\synch((a+\co{a}) \parallel a) = \synch(a+\co{a}) \oplus \synch(a) = \{\{a\},\{\co{a}\}\} \oplus \{\{a\}\} = \{\{a,\co{a}\}, \{a,a\}\}$$ We can see that there exists one synchronisation pair. Furthermore, if we consider the process $(b+\co{a}) \parallel a \parallel \co{b}$ we have:
\begin{align*}
&\synch((b+\co{a}) \parallel a \parallel \co{b}) = \synch(b+\co{a}) \oplus \synch(a) \oplus\synch(\co{b}) = \{\{b\},\{\co{a}\}\} \oplus \{\{a\}\} \oplus \{\{\co b\}\} = \\
&\{\{b,a\}, \{\co{a},a\}\}\oplus\{\{\co{b}\}\} =\{\{b,a,\co{b}\}, \{\co{a},a,\co{b}\} \}
\end{align*}
where there are two synchronisation pairs.
\begin{lem}
Given a reachable $\rtpl$ configuration $X$, then $X\not \fwd{\tau}$ is decidable. 
\end{lem}
\begin{proof}
By a simple induction on the structure of $X$.
\end{proof}

\section{Additional proofs of Section~\ref{sec:causality}}\label{sec:proofs}

\squareproperty*
\begin{proof}
The proof is by case analysis on the direction of the two transitions. We distinguish three cases according to whether the two transitions are both forwards, both backwards, or one forwards and the other backwards.

\begin{description}
\item[$t$ and $s$ forwards] 
  first we look at the case where the two actions are both communication actions. The proof is by induction on the structure of the common source configuration $X$. From $\nil$ no transition is possible hence this case can never happen. For a standard prefix, a single transition is possible, but for the choice of the key. Two transitions with the same key are not independent due to condition~\ref{confl:key} in Definition~\ref{def:confl}, hence there is nothing to prove. The cases of non-standard prefixes, timeout (both standard and non standard), hiding and constants follow by inductive hypothesis.
  If the configuration is a parallel composition then either we apply rule [\textsc{Par}] and its symmetric, or we apply rule [\textsc{Par}] and rule [\textsc{Syn}], or two [\textsc{Syn}].
  In the case of two applications of rule $[\textsc{Par}]$, if the same parallel component acts in both the cases, then the thesis follows from inductive hypothesis.
Otherwise we have
$X = X_1 \parallel X_2 \fwdns{\alpha_1\colorkey{i}} Y_1 \parallel X_2 = Y$ and $X_1 \parallel X_2 \fwdns{\alpha_2\colorkey{j}} X_1 \parallel Z_2 = Z$ with premises 
$X_1 \fwdns{\alpha_1\colorkey{i}} Y_1$ and $X_2 \fwdns{\alpha_2\colorkey{j}} Z_2$ (note that $i\not = j$ by case \ref{confl:key} of Definition~\ref{def:confl}).
By applying rule [\textsc{Par}] again we have
$Y = Y_1 \parallel X_2 \fwdns{\alpha_2\colorkey{j}} Y_1 \parallel Z_2$ and 
$Z = X_1 \parallel Z_1 \fwdns{\alpha_1\colorkey{i}} Y_1 \parallel Z_2$, as desired.
In the case of two applications of rule [\textsc{Syn}] we proceed by induction on the two components.
In the case of one [\textsc{Par}] and one [\textsc{Syn}] we proceed by induction on the component which acts in [\textsc{Par}].

In the case of choice, if the two transitions concern the same component then the thesis follows by inductive hypothesis. If the two transitions concern different components then they are not independent as they are in forward communication conflict, according to condition~\ref{confl:consumed} in Definition~\ref{def:confl}.

In the case $\pi_1 = \sigma$ and $\pi_2 = \sigma$ there is nothing to prove since then $Y=Z$ and $t=s$ by time determinism.

Finally, note that the case $\pi_1 = \alpha$ and $\pi_2 = \sigma$ is ruled out by Definition~\ref{def:confl} (first clause).
\item[$t$ and $s$ backwards] the case of two communication actions is
  similar to the previous one, noticing however that backward
  transitions are never in conflict according to
  Definition~\ref{def:confl}. Indeed, keys ensure that each component can
  take part in a single transition.

  The case of two $\sigma$ actions follows since time determinism
  holds also in the backward direction, since time actions are
  required on all components, and they need to have the same key.

  The case of a $\sigma$ action and a communication action follows
  from Proposition~\ref{prop:backalphasigma}.
  
\item[$t$ forwards and  $s$ backwards]
the case of communication actions is similar to the previous one: actions either are in conflict by condition \ref{confl:bkfwd} of Definition~\ref{def:confl} or they are generated by parallel components, hence can take place also in the opposite order.
  Two time transitions from the same configuration are always in conflict by condition \ref{confl:bkfwd} of Definition~\ref{def:confl}, since time actions are recorded in each component.
The case $\pi_1 = \alpha$ and $\pi_2 = \sigma$ (or vice versa) is analogous to the previous one. \qedhere
\end{description}
\end{proof}

\end{document}